\documentclass[acmlarge,authorversion]{acmart}

\usepackage{booktabs} 

\setcopyright{rightsretained}

\renewcommand\footnotetextcopyrightpermission[1]{} 
\usepackage{amsmath} 
\usepackage{amssymb}  
\usepackage{amsthm}
\usepackage{stmaryrd}
\usepackage{tikz,dsfont}
\usepackage{hyperref}
\usepackage[linesnumbered,ruled]{algorithm2e}
\usepackage{bbm}
\usepackage{graphicx,subfigure}
\usepackage[export]{adjustbox}
\usepackage{enumerate}

\usepackage{color-edits}
\addauthor{sb}{magenta}
\addauthor{ts}{blue}
\addauthor{ss}{green}
\addauthor{imp}{red}
\addauthor{ans}{teal}
\addauthor{remov}{orange}

\begin{document}
\title{The Price of Fragmentation in Mobility-on-Demand Services}

\author{Thibault S\'ejourn\'e}
\affiliation{%
  \institution{Ecole Polytechnique}
  \country{France}
}
\email{thibault.sejourne@polytechnique.edu}

\author{Samitha Samaranayake}
\affiliation{%
  \institution{Cornell University}
  \country{USA}
}
\email{samitha@cornell.edu}

\author{Siddhartha Banerjee}
\affiliation{%
  \institution{Cornell University}
  \country{USA}
}
\email{sbanerjee@cornell.edu}

\renewcommand{\shortauthors}{T. S\'ejourn\'e et al.}

\begin{abstract}
Mobility-on-Demand platforms are a fast growing component of the urban transit ecosystem.
Though a growing literature addresses the question of how to make individual MoD platforms more efficient, much less is known about the cost of market fragmentation, i.e., the impact on welfare due to splitting the demand between multiple independent platforms. Our work aims to quantify how much platform fragmentation degrades the efficiency of the system. In particular, we focus on a setting where demand is exogenously split between multiple platforms, and study the increase in the supply rebalancing cost incurred by each platform to meet this demand, vis-a-vis the cost incurred by a centralized platform serving the aggregate demand. 
We show under a large-market scaling, this {\em Price-of-Fragmentation} undergoes a phase transition, wherein, depending on the nature of the exogenous demand, the additional cost due to fragmentation either vanishes or grows unbounded. We provide conditions that characterize which regime applies to any given system, and discuss implications of this on how such platforms should be regulated.
\end{abstract}


\maketitle
\thispagestyle{empty}


\section{Introduction}
\label{sec:intro}

The ability to provide efficient, cost effective and sustainable urban transportation has been a major challenge for cities throughout the world in recent decades. This challenge is becoming even more difficult with the rapid urbanization we are observing at the moment - two thirds of the world population will live in urban areas by 2050~\cite{UN14}. Traditionally, urban mobility has been satisfied via either personal vehicle ownership or mass transit systems. Personal vehicles provide a very convenient means for door to door transportation, but induces high per-passenger negative externalities in terms of congestion, pollution and need for parking~\cite{TTI15}, making it a solution that does not scale well. Mass transit systems on the other hand are an efficient means to moving masses through regions of high spatio-temporal demand density, but are prone to first-and-last-mile problems.    

Mobility-on-Demand (MoD) systems (e.g. Uber, Lyft, Didi, Ola) provide a new alternative mode of transportation that preserves the convenience of personal vehicle ownership, while simultaneously reducing the burdens of car ownership. The growing popularity of these services are a testament to the perceived benefit of the systems by consumers. While the jury is out on whether these services currently help or hurt cities in terms of congestion and other externalities, they have the potential to help the overall efficiency of the transportation system as a whole, especially when these services are used to serve the first-and-last-mile for public transit and when shared rides occur (e.g. Lyft Line and UberPOOL).

The popularity of MoD systems has also given rise to a growing academic literature on these systems, in particular, with respect to fleet management, rider matching and pricing problems. However, there has been limited work in understanding the implications of multiple service providers (i.e., platforms) in the system. 
The economic benefits of competition for passengers and drivers (in terms of prices, wages earned, quality of service etc.) are well-established; sustaining this competition may however result in much higher overall operational costs.
In particular, the concern is that when multiple platforms fragment demand, and myopically optimize their own systems to serve their own customers, they may potentially result in a higher cost to society (in terms of pollution/fuel consumption/congestion, etc) as compared to a monopoly. 

Our work aims to understand and quantify the increased operational costs due to competition in MoD ecosystems. At the outset, we note that are two main challenges to doing so.
First, since MoD systems involve complex dynamics of many vehicles over large networks, understanding this cost requires careful modeling of demand/supply/rebalancing, and careful analysis that incorporates the underlying stochastic dynamics. Secondly, given the growth of MoD systems, the pertinent question not so much what the costs are currently, but rather, how they will behave as demand scales in the future; such counterfactual analysis is not directly accessible from the data. The techniques we develop address both these questions.

In more detail: we approach this problem by defining a novel notion of the \emph{Price of Fragmentation} -- a metric that captures how the efficiency of a MoD ecosystem degrades as demand is randomly split between platforms.
We focus on a setting where passengers are exogenously split between multiple platforms (i.e. we ignore endogenous competition effects between platforms), and study the increase in the system cost incurred by each platform to meet this demand, vis-a-vis the cost incurred by a monopolist platform serving the aggregate demand. Such a model allows us to study the scaling of the system costs as the demand for MoD systems grow. By doing so, we uncover a surprising phase-transition phenomena (see Figure.~\ref{fig:intro}), which we corroborate via synthetic experiments, as well as data-driven studies using the NYC taxi-cab data~\cite{donovan2014new}.

To summarize, our work is to the best of our knowledge the first study on the loss in efficiency (i.e. operational costs) of MoD systems due to the fragmentation of demand across multiple platforms. While further studies are needed to understand the effects of platform competition, we argue that understanding the systemic inefficiencies uncovered in our work is an important first step for guiding public policy for the MoD ecosystem.

\begin{figure}[t!]
	\centering
	\includegraphics[width=0.75\columnwidth]{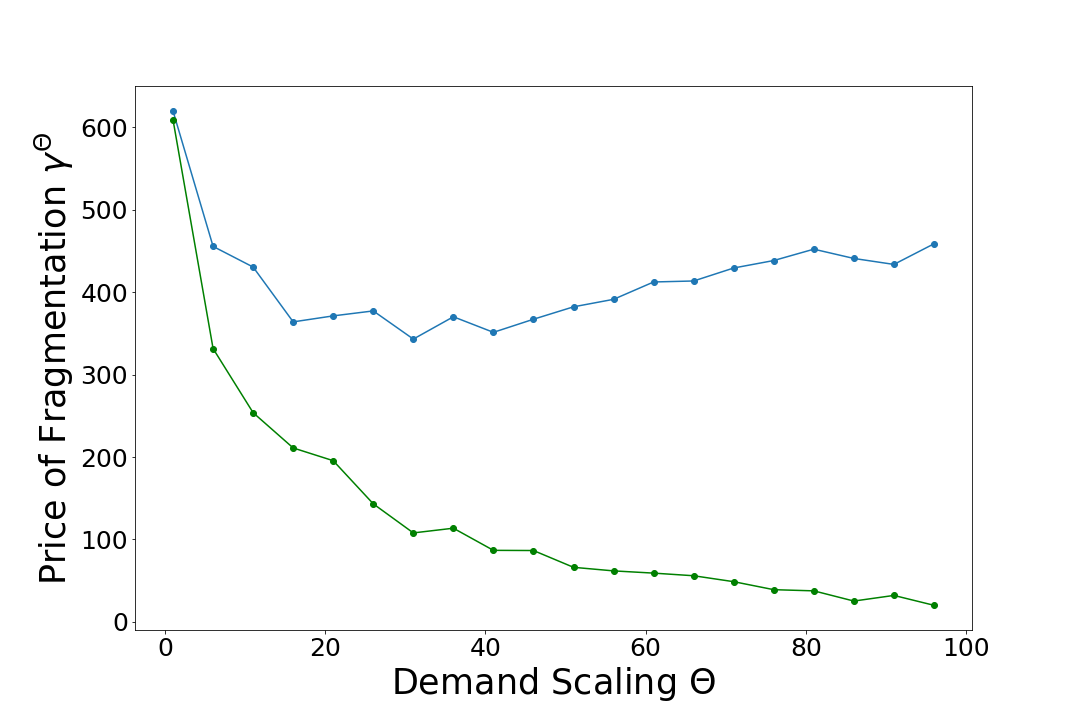}
	\caption{Simulation demonstrating the scaling behavior of the Price of Fragmentation (the difference between the rebalancing cost of the duopoly and the monopoly) in the fragmentation-resilient (green) and fragmentation-affected (blue) regimes. The demand distributions and costs are based the NYC taxi data; in particular, the two regimes shown here correspond to two successive hours (noon to 1 p.m and 1 to 2p.m.) on May 10. These time slots were chosen as their initial PoF (i.e., before scaling) are approximately equal, while our analysis shows that one is fragmentation-affected while the other is fragmentation-resilient.
		\vspace{-0.5cm}}
	\label{fig:intro}
\end{figure}

\subsection{Summary of our Results}

We consider an MoD system with either a monopolist, or two competing firms -- in the latter case, we assume the demand gets exogenously split between the two firms. Given this demand split, we focus on the \emph{cost of rebalancing}, i.e. the total cost of aligning vehicle supply with passenger demand. In this context, demand fragmentation leads to a total rebalancing cost in a duopoly that is necessarily greater than or equal to that of the monopolist cost. To quantify this, we consider a stochastic model of demand splitting, where we assume the demand for each origin-destination pair is split between the firms according to some random process; we define the {\em Price of Fragmentation} (PoF) to be the expected excess rebalancing cost incurred by the duopoly. To incorporate the growth of MoD systems, we study the PoF under a \emph{large-market scaling}: we scale all demands simultaneously by a factor $\theta$, before randomly splitting the demand between the firms (cf. Definition~\ref{def:POF} for details).

For ease of presentation, we initially focus on the case of homogeneous market shares, wherein the average fraction of demand going to each firm is the same across the network; we extend our results to heterogeneous market shares in Section~\ref{sec:general}.
It is important to note that our setup assumes an independent and sufficient vehicle supply, and exogenous demand splitting. In other words, we assume that consumers have idiosyncratic demand preferences and we do not model competition between the platforms. While we acknowledge that this is an important consideration for real world systems, the insights we gain under our simpler model are by themselves significant and non-obvious, and will likely continue to hold under more complex models as well. 

Our results are summarized as follows: for general networks under the large-market scaling, we show that the PoF undergoes a phase transition, depending on the structure of the underlying demand vector. In the first regime the PoF decays exponentially under scaling -- we refer to this as the \emph{fragmentation-resilient} regime. In contrast, in the second regime, the PoF grows under scaling, at a rate equal to the square root of the scaling parameter -- we refer to this as the \emph{fragmentation-affected} regime. Figure.~\ref{fig:intro} shows an example of these two cases, based on demand distributions estimated using the NYC taxi data over two different one-hour time slots; note that though both the settings appear to have identical PoF at the original demand, the two curves are markedly different as we scale the average demands. 

More specifically: the rebalancing cost in an MoD system is characterized by a minimum-cost circulation problem (cf. Equation~\eqref{eq:primal}), and our results aim to understand how the solution of this program changes under random demand splitting. Our main results (Theorems~\ref{thm:expdecay}, \ref{thm:sqroot} and \ref{thm:nonuniform}) together show that under homogeneous demand splitting, the \emph{system is in the fragmentation-affected regime if and only if the underlying demand vector is dual-degenerate with respect to the min-cost circulation program}, even when the market is inhomogeneously shared in the city. We also extend the results for non-homogeneous splits (Theorem~\ref{thm:nonuniform}), where the behavior may be more complex.

The above results, though providing an explicit characterization of the phase transition, are less interpretable as they are in terms of structural properties of the optima of an LP. We next provide an elegant combinatorial characterization for the fragmentation-affected regime (Theorem~\ref{thm:degen}). Our characterization suggests that the troubling case is when the demand induces \emph{local balanced clusters}, i.e., the network can be partitioned in two or more clusters such that there is (roughly) balanced underlying demand going to and fro between clusters, while the cost of traveling between clusters is high. We conjecture that the structure of urban traffic patterns can indeed lead to such balanced partitions~\footnote{A potential example is a distant airport, or twin city-centers, which typically have balanced incoming and outgoing flow of passengers over any sufficiently long interval.}. In Section~\ref{sec:sims}, we provide empirical evidence from the NYC taxi data, where we show that local balanced clusters do often arise in practice. We also perform simulations on synthetic examples, as well as on the NYC taxi data, to verify our PoF characterizations. Finally, in Section~\ref{sec:insights}, we discuss implications of our results on how MoD ecosystems should be regulated, and discuss possible extensions of our work.

\subsection{Related Work}

As alluded to before, MoD systems have recently gained wide interest in the research community. A large body of this work studies the design of optimal vehicle rebalancing policies, which we discuss below. Other topics of study which are less relevant to our work include dynamic pricing~\cite{banerjee2015pricing} and ride-pooling~\cite{SR14,SFR17}.

For rebalancing policies, a common approach is to formulate this as a control problem on a closed queuing network. 
Assuming that the number of cars and demand are very high, one can then use a fluid approximation to find rebalancing policies, and test them numerically afterwards~\cite{PSFR12,SGF15}. A related line of work is based on heuristics that enforce a certain fairness property. In transportation settings, George et al. used these to optimize weighted throughput~\cite{george2012stochastic}, Zhang et al. to minimize rebalancing costs~\cite{ZP14}. 

Most of the above papers are based on heuristics with limited guarantees.
More recently, several works \cite{ozkan2016dynamic,braverman2016empty} formally characterized the optimal control policy under fluid (or large-market) limits of closed queueing models of ride-sharing systems. Moreover, work by ~\citet{BFL17} provided a unified framework for designing rebalancing, dispatching and pricing policies, which provides approximation guarantees in finite settings, as well as obtaining the optimality of the above fluid and large-market limit policies via elementary arguments.

Finally we note that there is an extensive literature on competition between platforms~\cite{rochet2006two, rysman2009economics,weyl2010price}, including in network settings~\cite{ABH14,BEI14,AS15} and power markets~\cite{CBW17}. These are of less relevance to our work as we focus on the operational costs of demand fragmentation. However, it may be of interest to develop models synthesizing our work with such strategic competition settings -- we discuss this further in Section~\ref{sec:insights}.

\section{Setting and Preliminaries}
\label{sec:setting}

\subsection{Basic setting}

\noindent\textbf{Network Model}: We consider a city represented as a complete directed graph $G(V,E)$, comprising of $n$ nodes (or \emph{stations}) $V$ corresponding to neighborhoods, linked together by edges $E$ corresponding to {\em fastest routes} between any pair of stations. 
Each edge has an associated {\em average travel-time} $\tau_{ij}$, which can be interpreted as the cost incurred by a trip from $i$ to $j$ (in terms of driving time, or after appropriate re-weighting, miles driven/fuel consumption/pollution). Travel times need not be symmetric, but by definition must obey the triangle inequality.

\noindent\textbf{Demand Model}: Within this city, we study the operations of an MoD ecosystem (i.e., a one or more MoD providers) over a period of time during which the underlying demand for rides remains stationary~\footnote{Essentially we want a period over which the average demand is roughly constant -- for example, the morning rush hour on weekdays, or the afternoon hours on weekends, etc. Our results extend to time-varying arrivals, at the cost of additional notation.}. Since our focus in this work is on operational costs, we assume that this demand is exogenous (i.e., after pricing), and moreover, that all demand must be serviced with equal priority.

Over the period of time under consideration, we use $\{\Lambda_{ij}\}_{(i,j)\in E}$ to denote the \emph{demand vector}, where for each edge $(i, j)\in E$, $\Lambda_{ij}$ denotes the (random) number of customers requesting rides starting at node $i$ and terminating at node $j$. For any node $i$, we define $\Lambda_i = \sum_{j\in V} (\Lambda_{ji} - \Lambda_{ji})$ to denote the {\em net inflow} of passengers into node $i$. 
When clear from context, we sometimes abuse notation to use $\Lambda$ to denote the vector $\{\Lambda_i\}$. 


\noindent\textbf{The MoD Ecosystem}: To model fragmentation, we consider two models of the MoD ecosystem: $(i)$ a \emph{monopolist} setting, where all demand is serviced by a single MoD platform, and $(ii)$ a \emph{duopolist} setting, where we assume the demand is split between two firms. In particular, in the duopoly, for any edge $(i,j)\in E$, we assume the two firms $a$ and $b$ split the demand as $(\Lambda_{ij}^a,\Lambda_{ij}^b)$ such that $\Lambda_{ij}^a+\Lambda_{ij}^b = \Lambda_{ij}$ 
This split is assumed to be exogenously determined (i.e., unaffected by firms' strategies) -- in particular, we assume $\{\Lambda_{ij}^a,\Lambda_{ij}^b\}$ are randomly generated according to some specified distribution.

\noindent\textbf{Supply and Rebalancing}: Next we describe the supply side of the model for the monopolistic setting; in the duopolist setting, we assume the same model applies to each firm independently. This corresponds to a setting where each firm has an independent fleet within the time-period of interest; this is the case in some current markets (for example, where drivers are employed by the firms), and moreover, likely to be the case with self-driving fleets. Note this model does not capture {\em multi-homing drivers}, who work simultaneously for multiple firms; in Section~\ref{sec:insights}, we discuss how our results suggest that a sufficient amount of multi-homing may potentially help reduce overall social costs.

Our chief object of study is the {\em cost of rebalancing} $RC(\Lambda)$, i.e., the operational cost of rerouting empty supply to meet the demand. This given by the following {\em min-cost circulation LP}: 
\begin{align}
& \min \quad
\sum_{(i,j)} \tau_{ij} x_{ij} & \label{eq:primal}\\
& \mathrm{s.t.} \;\;\quad \sum_{j} (x_{ji} - x_{ij}) = \Lambda_i \quad\forall\, i\in V\quad;\quad 
\quad x_{ij}\geq 0 \quad \forall\, (i,j)\in E, \nonumber
\end{align}
where $\Lambda_i = \sum_{j} (\Lambda_{ij} - \Lambda_{ji})$ is the net inflow of demand at station $i$, and $x_{ij}$ is the number of rebalancing trips from station $i$ to $j$.
The rebalancing cost $RC(\Lambda)$ can also be expressed via its dual LP:
\begin{align}
\label{eq:dual}
& \max \qquad
\sum_{i} \alpha_{i}\Lambda_i \\
& \mathrm{s.t.} \qquad\quad \alpha_{i} - \alpha_{j} \leq \tau_{ij} & \forall (i,j)\in E\nonumber
\end{align}
The formulation of the rebalancing cost in terms of the min-cost circulation (or {\em optimal transport}) problem arises in many studies of MoD systems, and has different justifications based on different underlying assumptions. The critical point to observe is that $RC(\Lambda)$ is in a sense the best ex-post cost that a platform can obtain while satisfying all demand, and that dynamic rebalancing policies used in practice are known to match this bound very closely in various settings; we discuss these in more detail in Appendix~\ref{appsec:rebalancing}.
Consequently, we henceforth focus on the rebalancing cost $RC(\Lambda)$ as per the LP in \eqref{eq:primal} (alternately, the Dual LP~\eqref{eq:dual}), and ignore additional losses due to sub-optimality in dynamic scheduling.

Before proceeding, we present a structural characterization of the function $RC(\cdot)$, which we use extensively in what follows.
\begin{proposition}
	\label{prop:convex}
	The rebalancing cost $RC(\cdot)$ is convex and homogeneous (i.e., $\forall \; \theta \geq 0, \; RC(\theta\Lambda) = \theta RC(\Lambda)$).
\end{proposition}
\begin{proof}
	Both properties are most easily seen from the dual LP in \eqref{eq:dual}. Homogeneity follows immediately from factoring out $\theta$ from the objective; moreover, convexity follows from noting that $RC(\Lambda)$ can be written as the maximum of a finite number of linear functions $\alpha^\intercal \Lambda$, with $\alpha$ being the corner points of the dual polytope. 
\end{proof}

Finally, we note that though we focus on rebalancing cost in terms of time, an alternate commonly-used metric is the \emph{earth-movers distance} (EMD), which measures cost in terms of additional distance traveled for rebalancing. The two definitions are however equivalent up to scaling constants. In our work, we use
cost/travel-time/distance interchangeably to represent these costs.

\subsection{The Price of Fragmentation}
\label{ssec:pofdef}

Given the above setting, we can now define our measure for capturing the \emph{price of fragmentation} (PoF): the increase in rebalancing costs in a duopoly as compared to a monopoly, for serving the same demand with exogenous demand splitting. Such a notion captures the price paid by society in terms of additional vehicle miles in order to satisfy the same demand.

We first need some additional notation. Given a demand vector $\{\Lambda_{ij}\}$ in the monopolist setting, our aim is to understand the increase in rebalancing costs due to exogenous splitting of the demand into vectors $\{\Lambda_{ij}^a\}$ and $\{\Lambda_{ij}^b\}$ for firms $a$ and $b$ respectively. Moreover, to capture the growth of MoD platforms, we analyze the system under the so-called {\em large-market scaling}, wherein we scale all the incoming demands by a factor $\theta\in\mathbb{N}$.

Given the above setting, one way to model the price of fragmentation is to consider the worst case increase in rebalancing costs over all feasible demand splits. We discuss this notion in detail in Appendix~\ref{appsec:advpof}; however, under large-market scaling, it is easy to see that such a worst-case PoF must scale linearly in $\theta$ (this follows directly from the homogeneity of $RC(\Lambda)$). This suggests that for a meaningful analysis, we need a more refined notion.

We instead consider a \emph{stochastic model of demand splitting} between the firms. At a high level, such a model captures the idea that when users idiosyncratically choose between the two firms, then on a macro-level, the demand split between the firms in each location can be modeled as a random variable centered at the market share of the respective firms.
For ease of presentation, we assume that each $\Lambda_{ij}$ is an integer for every pair $(i,j)$, and the demand-splits between firms correspond to Binomial random variables with equal mean (i.e., homogeneous market share) across the network. {\em Our results however hold for general demand and demand-split distributions as well as heterogeneous market share} (cf. Section~\ref{sec:general}).

Given firms $a,b$, we define $\rho \in [0,1/2] $ to be the market share of firm $a$, and $1-\rho$ to be the market share of $b$. Consequently, for any edge $(i,j)\in E$, we assume that the demand $\Lambda_{ij}$ is split in two parts, with firm $a$ receiving $\Lambda_{ij}^a \sim \mbox{Binomial}\left(\Lambda_{ij}, \rho\right)$, and firm $b$ receiving $\Lambda_{ij}^b = \Lambda_{ij} - \Lambda_{ij}^a$.
We abuse notation to define $\Lambda^a \sim \mbox{Binomial} \left({\Lambda, \rho}\right)$ to be a vector such that each coordinate is an iid. rv. $\Lambda^a_{ij} \sim \mbox{Binomial} \left({\Lambda_{ij}, \rho}\right)$.
Now the increase in rebalancing cost is given by: 
\begin{align*}
\gamma \triangleq \mathbb{E}[RC(\Lambda^a) + RC(\Lambda^b)]-RC(\Lambda)
\end{align*}
Where $(\Lambda^a,\Lambda^b)$ denotes the random vector of demands for the two firms.
Now, under the \emph{large-market scaling} with demand splits $\Lambda^{a,\theta} \sim \mbox{Binomial} \left({\theta\Lambda, \rho}\right)$, and $\Lambda^{b,\theta} = \theta\Lambda - \Lambda^{a,\theta}$ (for $\theta\in\mathbb{N}$), we formally define the \emph{Price of Fragmentation} as follows:
\begin{definition}{(\textbf{Price of Fragmentation})}
\label{def:POF}
	Given monopolist demand $\Lambda$ and market-share $\rho$, the Price of Fragmentation is given by:
	$$
	\gamma^{\theta} \triangleq \mathbb{E}[RC\left(\Lambda^{a,\theta}\right) + RC\left(\Lambda^{b,\theta}\right)]-RC(\theta\Lambda)$$
\end{definition}
To get more insight into the PoF, we note first that given the homogeneity and convexity of $RC(\cdot)$, Jensen's inequality implies that $\gamma \geq 0$, i.e., competition necessarily reduces the system efficiency under stochastic splitting.
More precisely, noting that $\mathbb{E}[\Lambda^{a,\theta}] = \rho \theta \Lambda$, we get that:
\begin{align*}
\mathbb{E}[RC(\Lambda^{a,\theta}) + RC(\Lambda^{b,\theta})]
&\geq RC(\rho\theta\Lambda) + RC((1-\rho)\theta\Lambda) 
= RC(\theta.\Lambda)
\end{align*}
Note that since $RC(\cdot)$ is continuous and $\Lambda^{a,\theta}$ is binomial (and hence, has sub-Gaussian tails), the law of large numbers guarantees that $\gamma^{\theta}/\theta \rightarrow 0$.
Our results essentially characterize the rate of this convergence, and in the process, show that depending on the underlying demand vector, the stochastic PoF$(\theta)=\gamma^\theta$ either decays to $0$ or grows unbounded. We refer to the former as the \emph{fragmentation-resilient} regime, and the latter as the \emph{fragmentation-affected} regime.

To conclude this section, we make some observations:\\
-- From a technical perspective, we note that though we restrict to the case of a duopoly, our results generalize to any number of competitors. However, the monopoly to duopoly shift is in a sense the most drastic in terms of costs.\\
-- Similarly, our results also generalize to heterogeneous market shares, i.e., where the demand at different locations may split differently. Here though, in addition to our two regimes, we obtain a third wherein $\gamma^\theta$ scales linearly with $\theta$. We discuss this in more detail later in the paper.\\
-- Finally note that, in a sense, our model of exogenous demand splitting isolates the effects of firm-level competitive decisions (re. pricing, capacity etc.) in a manner similar to price-theoretic vs. microeconomic models for marketplaces~\cite{weyl2010price}. An advantage of this is that our results only depend on the structure of total demand, rather than particular assumptions of competitive dynamics. Consequently any global demand model accounting for competitive effects can be plugged in to extend the range of applications of our results.

\section{Warm up: Price of Fragmentation in a Two-Station Network}
\label{sec:2node}

To build some intuition into our main result, we first analyze a simple network with only two nodes. One advantage of this setting is that it allows us to solve for the rebalancing cost in closed form, thereby providing an analytic insight into the system behavior in different regimes. These insights prove useful in establishing the result for general networks in the next section.

In the case where our network only has two stations, we have the following system parameters:
\begin{itemize}
	\item Two stations $\{1,2\}$ with inter-station travel times $\tau_{12}=1$ and $\tau_{21} = \tau \geq 1$
	\item Monopolist demand vectors $\Lambda_{12} = \lambda$ and $\Lambda_{21} = \mu$,
	\item Optimal dual variables for each station under the monopolist demand vectors given by $\alpha_1 = 0$ and $\alpha_2$,
	\item Random demand-splits in the duopoly setting given by\\ $\Lambda_{12}^a \sim \mbox{Binomial} \left(\theta\lambda, \rho\right), 
	\Lambda_{12}^b = \Lambda_{12}-\Lambda_{12}^a$, \\
	$\Lambda_{21}^a\sim \mbox{Binomial} \left({\theta\mu, \rho}\right),\Lambda_{21}^b=\Lambda_{21} - \Lambda_{21}^b$
\end{itemize}
The setting is summarized in the following diagram:
\medskip
\begin{center}
	\begin{tikzpicture}
	\draw [->] (-2,0) -- (-0.5,0);
	\node at (-1.5,0.2) {$\Lambda_{12}=\lambda$};
	\draw [<-] (3.5,0) -- (5,0);
	\node at (4.5,0.2) {$\Lambda_{21}=\mu$};
	\draw (0,0) circle [radius=0.5];
	\node at (0,0) {1};
	\draw (3,0) circle [radius=0.5];
	\node at (3,0) {2};
	\draw [->] (0.433,0.25) -- (2.567,0.25);
	\node at (1.5,0.45) {$\tau_{12}=1$};
	\draw [<-] (0.433,-0.25) -- (2.567,-0.25);
	\node at (1.5,-0.5) {$\tau_{21}=\tau$};
	\end{tikzpicture}
\end{center}
\medskip
Note that our choice of $\tau_{12},\tau_{21}$ is without loss of generality (up to multiplicative scaling). Moreover, the dual constraint implies $-1 \leq \alpha_2 \leq \tau$. Consequently, the dual polytope has two corner points $(0, \tau)$ and $(0,\text{-}1)$, and: 
\begin{align*}
RC(\lambda, \mu) &= \max\{ \tau(\mu - \lambda), (\lambda - \mu)\}\\
&= \frac{(\tau -1)(\mu - \lambda) + (\tau+1)|\mu - \lambda| }{2}
\end{align*}
Given this explicit formula, we can now derive the asymptotic behavior of $\gamma^{\theta}$. To do so,we approximate the binomial distribution by a Gaussian distribution~\footnote{We make this assumption primarily to make the calculations more transparent, thereby giving some intuition for the general case. Note that this may introduce a bias as the demand may be negative; however, this is a second order effect, and is easily handled by our later theorems, which only require tail bounds.}; in particular, we assume the demand splits obey $\Lambda_{12}^a \triangleq X\sim \mathcal{N} \left({\rho\theta\lambda, \sigma_X^2}\right)$ and $\Lambda_{21}^a \triangleq Y \sim \mathcal{N} \left({\rho\theta\mu, \sigma_Y^2}\right)$ with $\sigma_X^2 = \theta\lambda\rho(1-\rho)$ and $\sigma_Y^2 = \theta\mu\rho(1-\rho)$. Now we have the following result
\begin{proposition}
	\label{prop:2node}
	For sufficiently large $\theta$, $\exists\,\, A_1,A_2\in \mathbb{R}$ such that:
	\begin{itemize}
		\item  $\lambda = \mu \; \Rightarrow \; \gamma^{\theta} = A_1\theta^{1/2}$
		\item  $\lambda \neq \mu \; \Rightarrow \gamma^{\theta} = A_2\theta^{-1/2}e^{-\frac{\rho(\lambda-\mu)^{2}}{2(1-\rho)(\lambda+\mu)}.\theta} + o\left(\theta^{-1/2}e^{-\frac{\rho(\lambda-\mu)^{2}}{2(1-\rho)(\lambda+\mu)}.\theta}\right)$
	\end{itemize}
\end{proposition}

The formal proof is presented in Appendix~\ref{appsec:proofs}; however, we can use this setting to point out some features which will help in understanding the general system in the next section: 
\begin{itemize}
\item Consider a particular realization of the demand splits -- now, if all three firms end up having to rebalance in the same direction (say from $1$ to $2$), then we get $0$ PoF for that sample. Note that this corresponds to all three LPs (for rebalancing under the monopolist demand, and individually for each  firm) attaining their maximum value at the same corner point of the dual polytope. In contrast, in the unbalanced case, the loss in efficiency comes from an inversion of the direction of rebalancing flow.
\item When $\lambda > \mu$, the difference of demand flows increases linearly with $\theta$, while the probability that $\lambda^a < \mu^a $ decays exponentially. This is the reason why we get an exponential decay in $\gamma^{\theta}$, and we see that the higher the discrepancy between input demands, the faster the convergence.
\item On the contrary, when $\lambda = \mu$, the decay is much slower because the total demand is perfectly balanced. Consequently, $\gamma^{\theta}$ accounts for the fluctuations of $X-Y$, which has variance on the order of $\sqrt{\theta}$. 
\end{itemize}

To understand the loss in efficiency more formally in a general setting, let's consider the dual LP (cf. \eqref{eq:dual}), and a realization of the demand-split $\Lambda^a=X$ and $\Lambda^b=\lambda-X$. Suppose the optimal dual corner point under this realization for the monopoly and each firm is given by $\alpha^a$, $\alpha^b$ and $\alpha^*$. Then we have:
\begin{align*}
RC(\Lambda^a)+RC(\Lambda^b)-RC(\Lambda)
=\Lambda^a . \alpha^a+ \Lambda^b . \alpha^b -\Lambda . \alpha^*
\end{align*}
If $\alpha^a=\alpha^b=\alpha^*$ then the loss is zero; inefficiency arises when the dual solutions are different. This suggests that the overall loss in rebalancing costs is proportional to the probability of having different dual solutions. We formalize this in the next section.

\section{Price of Fragmentation: The General Network Setting}
\label{sec:general}

The network with two nodes in the previous sections gives us important intuition to understand the behavior of the stochastic PoF in any general network with given demands and edge costs. In particular, Proposition~\ref{prop:2node} showed that the phase transition in the PoF $\gamma^{\theta}$ from exponential decay to square-root growth was induced by a change in the optimal corner point of the dual polytope. We will now extend this property to general networks.

We start by defining some mathematical objects we use in our further results. First, we define $\mathcal{E}$ to be the set of corner points of the dual polytope (cf. equation~\eqref{eq:dual}); note that the rebalancing cost $RC(\Lambda)$ is the maximum inner product of $\Lambda$ over all corner points in $\mathcal{E}$. Next, for each $\alpha \in \mathcal{E}$, we define the following set:
$$C_{\alpha} = \{\; \lambda \ s.t. \ RC(\lambda) = \alpha^\intercal\lambda \; \} $$
\begin{proposition}
	$\forall \; \alpha \in \mathcal{E}$, $C_{\alpha}$ is a closed cone. Furthermore, the intersection of two cones $(C_{\alpha},C_{\alpha'})$ is a hyperplane.
\end{proposition}
\begin{proof}
	It is immediate that $C_{\alpha}$ is a cone due to the homogeneity of $RC(\cdot)$. Secondly, if we define $h(\lambda)=RC(\lambda)-\alpha^\intercal\lambda$, then it is clear that $h$ is continuous and we have $C_{\alpha}=h^{-1}(\{0\})$ -- this verifies that it is closed. Finally, for distinct $(\alpha,\beta) \in \mathcal{E}$ consider $\lambda \in C_{\alpha}\cap C_{\beta}$ -- we have $RC(\lambda)=\alpha^\intercal\lambda=\beta^\intercal\lambda$, thus $(\alpha - \beta)^\intercal\lambda=0$, which is the equation of a hyperplane.
\end{proof}
These properties now allow us to derive a nice geometric interpretation of what happens when we scale the demand to infinity; we depict this in Figures.~\ref{fig:2} and \ref{fig:3} respectively for the fragmentation-resilient and fragmentation-affected regimes.
The solid circle shows the average demand $\mathbb{E}[\Lambda^{a,\theta}] = \rho\theta\Lambda$, which moves along a line as we scale $\theta$; the circles depict a high-probability envelope for the (normalized) realized demand $\Lambda^{a,\theta}/\theta$. The radius of the envelope decreases as we rescale by $\theta$.

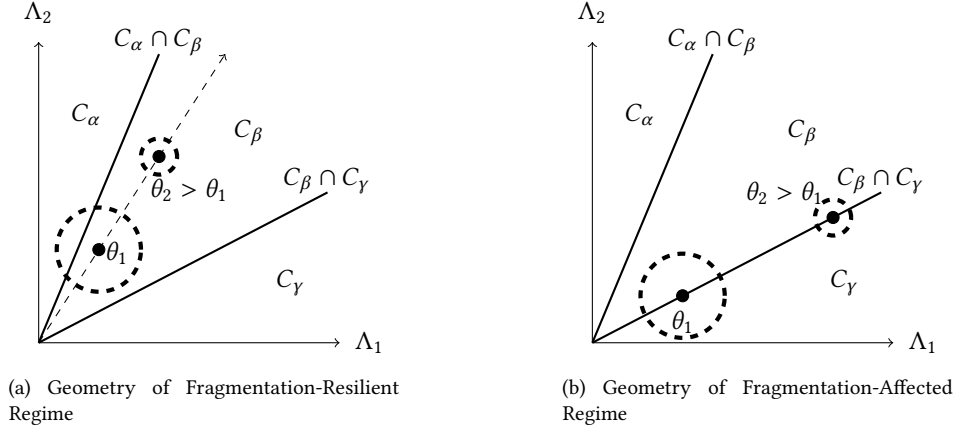
\begin{figure*}[t]
\centering
\subfigure[Geometry of Fragmentation-Resilient Regime]{
\label{fig:2}
		\begin{tikzpicture}[scale=0.8]
		\draw [->] (0,0) -- (5,0);
		\draw [->] (0,0) -- (0,5);
		\draw [thick] (0,0) -- (2,4.8);
		\draw [thick] (0,0) -- (4.8,2.5);
		\draw [->, dashed] (0,0) -- (3.1,4.8);
		\draw [dashed, ultra thick] (1,1.5484) circle [radius=0.7];
		\draw [dashed, ultra thick] (2,3.0968) circle [radius=0.3];
		\draw [fill] (1,1.5484) circle [radius=0.1];
		\draw [fill] (2,3.0968) circle [radius=0.1];
		\node at (5.5,0) {$\Lambda_{1}$};
		\node at (0,5.5) {$\Lambda_{2}$};
		\node at (0.8,3.8) {$C_\alpha$};
		\node at (3.5,3.5) {$C_\beta$};
		\node at (4.2,1) {$C_\gamma$};
		\node at (2,5) {$C_{\alpha}\cap C_{\beta}$};
		\node at (4.8,2.7) {$C_{\beta}\cap C_{\gamma}$};
		\node at (1.3,1.5) {$\theta_1$};
		\node at (2.5,2.55) {$\theta_2 > \theta_1$};
		\end{tikzpicture}
}
\hspace{2cm}
\subfigure[Geometry of Fragmentation-Affected Regime]{
\label{fig:3}
		\begin{tikzpicture}[scale=0.8]
		\draw [->] (0,0) -- (5,0);
		\draw [->] (0,0) -- (0,5);
		\draw [thick] (0,0) -- (2,4.8);
		\draw [thick] (0,0) -- (4.8,2.5);
		\draw [dashed, ultra thick] (1.5,0.78125) circle [radius=0.7];
		\draw [dashed, ultra thick] (4,2.0833) circle [radius=0.3];
		\draw [fill] (1.5,0.78125) circle [radius=0.1];
		\draw [fill] (4,2.0833) circle [radius=0.1];
		\node at (5.5,0) {$\Lambda_{1}$};
		\node at (0,5.5) {$\Lambda_{2}$};
		\node at (0.8,3.8) {$C_\alpha$};
		\node at (3.5,3.5) {$C_\beta$};
		\node at (4.2,1) {$C_\gamma$};
		\node at (2,5) {$C_{\alpha}\cap C_{\beta}$};
		\node at (4.8,2.7) {$C_{\beta}\cap C_{\gamma}$};
		\node at (1.5,0.35) {$\theta_1$};
		\node at (3.2,2.5) {$\theta_2 > \theta_1$};
		\end{tikzpicture}
}
\caption[Geometric intuition behind the two PoF regimes]{In Fig.~\ref{fig:2} we demonstrate the geometry of the \emph{fragmentation-resilient regime} (i.e., where $\gamma^\theta$ undergoes exponential decay), while in Fig.~\ref{fig:3} we do the same for the \emph{fragmentation-affected} regime (i.e., where $\gamma^\theta=\Omega(\sqrt{\theta})$). In both plots, the straight dashed line represents the expected demand for the first firm as demand scaling $\theta$ increases, while the two solid circles each represent the expected demand under two scaling values $\theta_1 < \theta_2$. Note that the expected demand line is within the dual cone in Fig.~\ref{fig:2}, while it is on the cone boundary in Fig.~\ref{fig:3}.
The dotted circles around the expected value represent a high-confidence ball for the realized demands; the radii of the balls are based on the variance of the rescaled demand r.v. $\lambda^{\theta}/\theta$, which decreases as we scale $\theta$. 
Due to random fluctuations, the probability of having an extra loss is high in Fig.~\ref{fig:3}, inducing a loss of the same order of the fluctuations.}
\label{fig:2regimes}
\end{figure*}

In the fragmentation-resilient setting, the expected demand remains on a line inside a given dual cone, while in the fragmentation-affected regime, it lies on a hyperplane defined by the intersection of two cones. 
Since $\rho$ is homogeneous, the expected demands after splitting for both firms in the duopoly also remain on the same line upon scaling. The true demands however have random fluctuations around its expectation, but since the cone becomes wider while scaling, while the demand distributions concentrate (assuming well-behaved tails), the probability of escaping the cone decays fast. 

On the other hand, if the optimal dual solution is degenerate, then the expected regime lies in the intersection of 2 or more cones. Now under scaling the expected monopolist demand vector remains on the hyperplane, but despite concentrating about its expected value, the probability of having each duopolist firm's realized demand-split fall on either side of the boundary remains significant.

We now state our main results, which formalize these intuitions. Hereafter, for any set $\mathcal{A}$, $\mathring{\mathcal{A}}$ denotes its interior. First, we characterize the exponential decay regime:
\begin{theorem}
	\label{thm:expdecay}
	(\textbf{Fragmentation-resilient regime}) Given network $G$, monopolist demand vector $\Lambda$ and demand-split $\rho$, suppose that the following conditions hold:
	\begin{itemize}
		\item $\Lambda \in \mathring{C}_{\beta}$ for some $\beta \in \mathcal{E}$ (i.e., the monopolist demand vector lies in the interior of a dual cone)
		\item There is a random vector $Z$ such that $\Lambda^{\theta,a} = \rho\theta\Lambda + \sqrt{\theta} Z$ with $\mathbb{E}[Z_e]=0\, \; \forall\,e\in E$ and $\mathbb{E}[||Z||_1]<\infty$
		\item There is a function $f$ such that $\forall\,e\in E$, we have $\mathbb{P}[|Z_{e}|>t]=\mathcal{O}(f(t))$ 
	\end{itemize}
	Then $\gamma^{\theta}=\mathcal{O}\left(\theta f\left(\sqrt{\theta}\right)\right)$
\end{theorem}
\noindent In the particular case of binomial demand splits, note that as long as $Z_i$ has sub-Gaussian tails, $\gamma^\theta = O(\theta e^{-c\theta})$ for some constant $c$.
\begin{proof}
	Define $x=\mathbb{E}[\Lambda^{\theta,a}]=\rho\Lambda$ and $y=\mathbb{E}[\Lambda^{\theta,b}]=(1-\rho)\Lambda$. As we said earlier, the boundaries between dual cones are hyperplanes $\mathcal{P}_{(\alpha,\beta)},\; \forall \: (\alpha,\beta) \in \mathcal{E}$. Thus we can define the orthogonal projections of $x$ and $y$ on each of those boundaries. Since there is a finite number of corner points, a strictly positive minimum distance is reached over all planes. We define the minimum distance between the expected demand vectors of both companies and the cone's boundary as: 
	\begin{align*}
	\delta=\min\left\{\underset{\alpha \in \mathcal{E}}{\min}\; d(x,\mathcal{P}_{(\alpha,\beta)}), \; \underset{\alpha \in \mathcal{E}}{\min}\; d(y,\mathcal{P}_{(\alpha,\beta)}) \right\}
	\end{align*}
	Moreover, since $\Lambda \in \mathring{C}_{\alpha}$, we have that the Euclidean balls around these points are contained in the cone $C_\alpha$, i.e., $B(x,\delta)\subset C_{\alpha}$ and $B(y,\delta)\subset C_{\alpha}$. Note that for any demand-split, we have $\lambda \in B(x,\delta)$ if and only if $\Lambda-\lambda \in B(y,\delta)$. 
	
	When we scale the demand vectors, since we use orthogonal projections and the distance is homogeneous, then $\delta$ is also homogeneous with respect to $\theta$. Thus, if our random demand-splits lie within a $\theta\delta$ ball around their expected values, we preserve the same corner point and undergo no loss in efficiency. To find an upper bound on the loss of efficiency when our random variable is outside those spheres, we define $\bar{\beta}\triangleq\underset{\beta \in \mathcal{E}}{\max}\|\beta\|_{\infty}$.
	
	We also denote $\alpha_1$, $\alpha_2$ and $\beta$ the corner points such that $RC(\Lambda^{\theta,a})=\alpha_1^\intercal\Lambda^{\theta,a}$, $RC(\theta\Lambda-\Lambda^{\theta,a})=\alpha_2^\intercal(\theta\Lambda-\Lambda^{\theta,a})$, and $RC(\theta\Lambda)=\beta^\intercal\theta\Lambda$. What is important to notice is that the corner points $\alpha_1$ and $\alpha_2$ implicitly depend on $Z$ and are consequently random variables. 
	
	Now, we can bound the normalized loss in rebalancing costs under scaling as follows:
	\begin{align*}
	\gamma^{\theta} &= \mathbb{E}\left[\alpha_1^\intercal\Lambda^{\theta,a}+\alpha_2^\intercal(\theta\Lambda-\Lambda^{\theta,a})- \beta^\intercal\theta\Lambda\right]\\
	&=\theta(\rho\alpha_1+(1-\rho)\alpha_2-\beta)^\intercal\Lambda  + \sqrt{\theta}\mathbb{E}\left[(\alpha_1-\alpha_2)^\intercal Z\right]\\
	&\leq 2\theta\bar{\beta}||\Lambda||_1\, \mathbb{P}\left[\exists\, e\in E\mbox{ s.t.}\; \sqrt{\theta}|Z_e| > \theta\delta\right] + 
	\sqrt{\theta}||\alpha_1-\alpha_2||_{\infty}\mathbb{E}\left[\ ||Z||_{1}.\mathds{1}_{\left\{\exists\, e\in E\mbox{ s.t.}\; \sqrt{\theta}|Z_e | > \theta\delta\right\}}\right]\\
	&\leq \left(2\theta\bar{\beta}||\Lambda||_1 + 2\bar{\beta}\sqrt{\theta}\mathbb{E}\left[\ ||Z||_{1} \; \big|\;\exists\,\, e\in E \mbox{ s.t.}\; |Z_e| > \sqrt{\theta}\delta\right]\right)
	\times \mathbb{P}\left[\exists\, e\in E\mbox{ s.t.}\; |Z_e| > \sqrt{\theta}\delta\right]
	\end{align*}
	The first inequality follows from applying H\"{o}lder's inequalities, and the requirement that for positive PoF, the demand-split vector must lie in a cone other than $\mathring{C}_{\beta}$. Now we want to prove that the expectation inside the brackets in the last inequality is finite. For this, recall we assumed that $\mathbb{E}[||Z||_{1}]$ is finite; furthermore, due to the law of total expectation, and because the expectations are positive, the finiteness of $\mathbb{E}[||Z||_{1}]$ implies the finiteness of any conditional expectation on $||Z||_{1}$, and thus the one we want. The only remaining point is to upper bound the probability of having $Z$ out of the ball. Via a union bound, we get:
	\begin{align*}
	\mathbb{P}\left[\exists\, e\in E \mbox{ s.t.} \; |Z_e| > \delta \sqrt{\theta}\right]
	\leq N^2\mathbb{P}\left[ |Z_e| > \delta\sqrt{\theta}\right]
	= \mathcal{O}(f(\sqrt{\theta}))
	\end{align*}
	Combining the inequalities, we get the desired result:
	$$\gamma^{\theta}=\mathcal{O}\left(\theta f(\sqrt{\theta})\right)$$
\end{proof}


Similarly, we can also characterize the square-root decay regime:
\begin{theorem}
	\label{thm:sqroot}
	(\textbf{Fragmentation-affected regime}) Given network $G$, monopolist demand vector $\Lambda$ and demand-split $\rho$, suppose that the following conditions hold:
	\begin{itemize}
		\item $\exists\, \alpha \neq \beta \in \mathcal{E}$ such that $\Lambda \in C_{\alpha}\cap C_{\beta}$ (i.e., the monopolist demand vector lies in the hyperplane defined by the intersection of two or more dual cones)
		\item There is a r.v. $\xi$ such that $\Lambda^{\theta,a} = \rho\theta\Lambda + \sqrt{\theta}Z$ with $\mathbb{E}[Z_e]=0\, \; \forall\,e\in E$ and $\mathbb{E}[||Z||_1]<\infty$
		\item $\forall\,e\in E$, we have $\mathbb{P}\left[|Z_e|>t\right]=\mathcal{O}\left(f(t)\right)$ for some $f(t)=\mathcal{O}(t^{-1})$
	\end{itemize}
	Then $\gamma^{\theta}=\Theta\left(\theta^{1/2}\right)$
\end{theorem}
\begin{proof}
	First we show the lower bound. Note that in this setting, we have $RC(\Lambda)=\alpha^\intercal\Lambda=\beta^\intercal\Lambda$ due to the hypothesis of lying on the interface of two cones. Furthermore, $\alpha^\intercal\Lambda^{\theta,a} \geq \beta^\intercal\Lambda^{\theta,a}$ if and only if $\beta^\intercal(\theta\Lambda-\Lambda^{\theta,a}) \geq \alpha^\intercal(\theta\Lambda-\Lambda^{\theta,a})$. Due to this, it is sufficient to separate two cases depending on which of these two corner points gives the best value. Furthermore, ignoring all other corner points gives a lower bound because the definition of the cones is such that changing to another corner point means that this third point yields a higher score. Thus we have:
	
	\begin{align*}
	\gamma^{\theta}&=\mathbb{E}\left[RC(\Lambda^{\theta,a})+RC(\theta\Lambda-\Lambda^{\theta,a})-RC(\theta\Lambda)\right]\\
	&\geq \mathbb{E}\big[(\alpha^\intercal\Lambda^{\theta,a}+\beta^\intercal(\theta\Lambda-\Lambda^{\theta,a})-\theta\beta^\intercal\Lambda).\mathds{1}_{\alpha^\intercal Z \geq \beta^\intercal Z}+
	(\beta^\intercal\Lambda^{\theta,a}+\alpha^\intercal(\theta\Lambda-\Lambda^{\theta,a})-\theta\alpha^\intercal\Lambda).\mathds{1}_{\beta^\intercal Z \geq \alpha^\intercal Z}\big]\\
	&= \theta^{1/2}.\mathbb{E}\left[(\alpha-\beta)^\intercal Z\left(\mathds{1}_{\alpha^\intercal Z \geq \beta^\intercal Z}-\mathds{1}_{\beta^\intercal Z \geq \alpha^\intercal Z}\right)\right]\\
	&= \theta^{1/2}\mathbb{E}\left[\left|(\alpha-\beta)^\intercal Z\right|\right]= \Omega(\theta^{1/2})
	\end{align*}
	
	Thus the first part of the proof is finished.
	
	Next we focus on the upper bound of $\gamma^{\theta}$. The proof will reuse ideas of the previous proof. In fact, let's separate $\alpha$ and $\beta$ of the other corner points. The minimum radius $ \delta$ such that a ball of radius $\delta$ can fit in $C_{\alpha}\cup C_{\beta}$ is well defined and is strictly positive. Thus, we can separate the rebalancing loss in two parts such that:
	\begin{align*}
	\gamma^{\theta}
	&=\mathbb{E}\Big[\left(RC(\Lambda^{\theta,a})+RC(\Lambda^{\theta,b})-RC(\theta\Lambda)\right)
	\left(\mathds{1}_{\{\forall\, e\in E,\,  Z_e \leq \sqrt{\theta}\delta\}} +\mathds{1}_{\{\exists\, e\in E \mbox{ s.t.} \, Z_e > \sqrt{\theta}\delta\}}\right)\Big]\\
	&\leq \theta^{1/2}.\mathbb{E}\left[\big|(\alpha-\beta)^\intercal Z\big|\mathds{1}_{\{\forall\,e\in E,\, Z_e \leq \sqrt{\theta} \delta\}}\right] +\mathcal{O}(f(\sqrt{\theta}))\\
	&\leq \theta^{1/2}.\mathbb{E}\left[\big|(\alpha-\beta)^\intercal Z\big|\right] +\mathcal{O}(f(\sqrt{\theta}))
	= \mathcal{O}(\theta^{1/2})
	\end{align*}
	If there are more than two corner points the proof is similar but has to be slightly modified. The lower bound is the same by considering only two of those points. For the upper bound, there are still two parts by defining the ball fitting in all cones. One part decays exponentially, and the second one can be upper bounded by summing the absolute value of the random variable $(\alpha-\beta)^\intercal Z$ over all pairs of corner points which are optimal for $\Lambda$.
\end{proof}

We state Theorems \ref{thm:expdecay} and \ref{thm:sqroot} in greater generality than required, so as to admit different distributions for random demand-splitting in the duopolist setting. 
For the setting of constant $\Lambda$ and binomial splitting we defined earlier, we have that $Z$ is sub-Gaussian, i.e., $\forall\,e\in E$, we have $\mathbb{P}[|Z_e|>t]=\mathcal{O}(e^{-at^2})$ for some constant $a$; thus we get an exponential decay when the demands are unbalanced. 
If instead $\Lambda_{ij}$ comes from a Poisson process with mean $\mathbb{E}[\Lambda_{ij}] = \lambda_{ij}$, and furthermore, under scaling, its expectation is $\theta\lambda_{ij}$. In this case, assuming probabilistic demand splitting, the random demand experienced by each firm follows a Poisson distribution with mean and variance linear in $\theta$; note that this distribution has sub-Exponential tails, i.e.,  $\forall e\in E$, we have $\mathbb{P}[|Z_e|>t]=\mathcal{O}(e^{-at})$ for some constant $a$. The decay in $\gamma^{\theta}$ is now slower than the binomial case; in particular, when the demand is unbalanced, we have $\gamma^{\theta}=\mathcal{O}(\theta e^{-a\sqrt{\theta}})$.

Till now we considered the case of homogeneous market share $\rho$; however, we can easily extend our results to the case of inhomogeneous market shares percentage $\rho$. This generalization is interesting as market share can be neighborhood-dependent in practice. For instance, different parts of the city support different demographics, and depending on this, one firm might be preferred over the second one due to its quality of service or cheaper price. To capture this, we consider $\rho=\{\rho_{ij}\}$ to be a vector, and define the demand-split for firm $a$ as:
$$\mathbb{E}[\Lambda^{\theta,a}]=\theta\rho\odot\Lambda \qquad  \qquad \mathbb{E}[\Lambda^{\theta,b}]=\theta(1-\rho)\odot\Lambda$$
where $\odot$ is the Schur (or element-wise) product of $\rho$ and $\Lambda$.  

Due to the inhomogeneity in splitting, we can only take advantage of the homogeneity of $RC(\cdot)$ with $\theta$. Geometrically, this means that the expectation of the split demands no longer lie on the same line as the monopolist demand. Now let $\alpha$, $\beta$ and $\eta$ denote the optimal corner points of the dual polytope for demand vectors $\rho\odot\Lambda$, $(1-\rho)\odot\Lambda$ and $\Lambda$ respectively -- note that unlike before, these corner points can now be different. However, we  can still characterize the PoF as before, via the following theorem: 
\begin{theorem}
	\label{thm:nonuniform}
	Given network $G$, demand vector $\Lambda$ and demand-split vector $\rho$, assume the following hold:
	\begin{itemize}
		\item $\Lambda \in C_{\eta}$ for some $\eta \in \mathcal{E}$ (i.e., the monopolist demand vector lies in the interior of a dual cone)
		\item There is a r.v. $Z$ such that $\Lambda^{\theta,a} = \theta\rho\odot\Lambda + \sqrt{\theta}Z$ with $\mathbb{E}[Z_e]=0\,\forall\,e\in E$ and $\mathbb{E}[||Z||_1]<\infty$
		\item $\forall\,e\in E$, we have $\mathbb{P}[|Z_e|>t]=\mathcal{O}(f(t))$ with $f(t)=\mathcal{O}(t^{-1})$
	\end{itemize}
	Now suppose we define: 
	$$ L=\alpha^\intercal\rho\odot\Lambda+\beta^\intercal(1-\rho)\odot\Lambda-\eta^\intercal\Lambda$$
	Then we have: 
	\begin{itemize}
		\item $\rho\odot\Lambda \in \mathring{C}_{\alpha} \; and \; (1-\rho)\odot\Lambda \in \mathring{C}_{\beta}\Rightarrow \gamma^{\theta}= L\theta +\mathcal{O}(\theta f(\sqrt{\theta}))$
		\item $\rho\odot\Lambda \in C_{\alpha}\setminus\mathring{C}_{\alpha} \; or \; (1-\rho)\odot\Lambda \in C_{\beta}\setminus\mathring{C}_{\beta} \Rightarrow \gamma^{\theta} = L\theta + \Theta(\theta^{1/2})$
	\end{itemize}
\end{theorem}
\begin{proof}
The proof is the along the same lines as the homogeneous case, except that each companies' bounds have to be treated separately. The details are given in Appendix~\ref{appsec:proofs}.
\end{proof}
Taken together, the above results thus completely characterize the decay rates of $\gamma^\theta$. Note that in the general heterogeneous case, we actually get three regimes for $\gamma^{\theta}$: if $L=0$, then we get the exponential decay and square root growth as before; in addition, if $L>0$, then we now get a \emph{linear growth regime}.
Also, note that all the above proofs treats each company separately and then sums to get the required bounds. Thus, the results can easily be adapted for more than 2 firms, leading to similar conclusions. For sake of simplicity we stick to the case of 2 firms, but the general result is presented as Theorem~\ref{thm:anycompany} in Appendix~\ref{appsec:proofs}.

The above results thus indicate that the PoF undergoes a phase transition, going from converging to $0$ to diverging to $\infty$ as the monopolist demand vector $\Lambda$ approaches a boundary hyperplane between two dual cones. From an operational point of view, it is thus critical to understand conditions for diverging PoF; in particular, we desire conditions which are more intuitive rather than our characterization in terms of dual degeneracy. We now try to obtain such a condition.

Consider a setting with demands $\lambda_i=\sum_j\lambda_{ji}-\lambda_{ij}$, and let $(x^*, \; \alpha^*)$ denote a pair of optimal solutions for the primal and dual respectively; then by complementary slackness, we have that $x_{ij}^*>0 \Rightarrow \alpha_i^*-\alpha_j^*=\tau_{ij}$.
Now consider subgraphs $H$ and $\hat H$ be two subgraphs defined by the set of all nodes $V$ and edges:
\begin{align*}
E_{H}=\{(i,j) \; | \; x_{ij}>0\}\qquad \qquad
E_{\hat H}=\{(i,j) \; | \; \alpha_i^*-\alpha_j^*=\tau_{ij}\}
\end{align*}
$H$ thus represents the support of positive rebalancing flows for the monopolist setting; from the non-negativity of $\tau_{ij}$, it is easy to observe that $E_H$ forms a directed acyclic graph.
$\hat H$ is the subgraph defined by tight dual constraints. By complementary slackness, we have $E_{H} \subset E_{\hat H}$. 
We now have the following theorem that characterizes dual degeneracy in the rebalancing cost LP \eqref{eq:primal}.
\begin{theorem}
	\label{thm:degen}	
	For the optimal rebalancing solution to be dual degenerate, a necessary condition is that $E_H$ contains at least two connected components. Moreover, this is sufficient with probability $1$ under small random perturbations of the edge-weights $\tau_{ij}$.
\end{theorem}
\begin{proof}
	First, suppose $E_H$ has a single connected component. Now by complementary slackness, the dual variables $\alpha$ are uniquely defined by: $\forall (i,j)\in E_H, \; \alpha_i-\alpha_j=\tau_{ij}$. 
	This shows the necessity of our condition.

	Next, note that under random perturbations of the edge weights $\tau_{ij}$, we have $E_{H}=E_{\hat H}$ with probability $1$. Now suppose that $E_{H}$ (and hence $E_{\hat H}$) comprises of multiple connected components. We now show that we can extend it without modifying the LP value or the saturated inequalities.

	Let $A$ be one of the connected components and $\delta(A)$ denote the edges in the cut of $A$; we define $\epsilon=\inf_{(i,j)\in \delta{(A)}} |\alpha_i-\alpha_j-\tau_{ij}| >0$ (by definition of $E_{\hat H}$). 
	We now consider a new dual corner point defined by $\beta_i=\beta_i+\epsilon$ if $i\in A$ and $\beta_i=\alpha_i$ otherwise. This has no incidence with respect to the saturated inequalities, and it will saturate one inactive inequality thanks to its definition. If (i,j) are in A then the argument of the uniform translation applies (see proposition \ref{prop:degen} in Appendix~\ref{appsec:proofs}). Consequently, this new edge is now active and connects A to another connected component. Furthermore, the score for $\beta$ is the same as for $\alpha$ since we have:
	\begin{align*}
	\sum_i \alpha_i \lambda_i&=\sum_{(i,j)} (\alpha_i-\alpha_j)x_{ij}^*
	=\sum_{(i,j)\in E_H} (\alpha_i-\alpha_j)x_{ij}^*
	\end{align*}
	Thus $\alpha_i-\alpha_j\neq\beta_i-\beta_j \; \Leftrightarrow x_{ij}^*=0$, and this modification has no impact on the score.
	Finally, since we have more constraints than necessary to uniquely characterize a corner point, this means that two corner points are optimal for the dual LP.
\end{proof}
Theorem~\ref{thm:degen} gives some intuition into when dual degeneracy happens; more importantly, it gives an easily testable condition. Note though that for a setting to be fragmentation-affected, we need to display two disconnected components in the rebalancing flows, but which has non-zero demand flowing between the components (otherwise, the underlying demand itself is disconnected!).
We next use this to validate our results based on simulation of synthetic and real-world settings.

\section{Simulations and Experiments}
\label{sec:sims}

We first use synthetic experiments to numerically verify our phase transition theorems; for this, we ran simulations on a simple four nodes network, which allows us to easily create instances which are fragmentation-resilient and fragmentation-affected. Next, we look at the NYC taxi data, and use Theorem~\ref{thm:degen} to classify the fraction of instances falling in each class -- surprisingly, we detect the fragmentation affected on a significant fraction of the data. We also perform several robustness checks to validate our findings.

\subsection{Synthetic Experiments}
\label{ssec:synthexp}

First, we experimentally study our characterization of the phase transition in the PoF in terms of dual degeneracy of $RC(\Lambda)$. Although the two-nodes network we used in Section \ref{sec:2node} already allowed us to exhibit both regimes, it has a very simple degenerate subset wherein $\Lambda_{ij}=\Lambda_{ji}$. In order to build more complex synthetic examples that exhibit both regimes, we consider a four node network depicted in Fig.~\ref{fig:4node}, comprising of $4$ nodes, with inter-node average travel-times as indicated~\footnote{We only indicate travel-times between nodes $(1,3),(1,4),(2,3)$ and $(2,4)$ as these are the only ones relevant to the circulation LP for the demands we consider.}. By varying the monopolist demands between the nodes, we use the network to verify the conditions for the two regimes given in Theorem~\ref{thm:degen}.

Recall that $\Lambda_i$ denotes the net {\em inflow} of demand into node $i$. To generate instances corresponding to the two regimes, we consider two different demand vectors, $\Lambda^1 = [2,3,-4,-1]$ and $\Lambda^2 = [2,3,-3,-2]$; note that this is sufficient for solving $RC(\Lambda)$, and we do not need to specify $\Lambda_{ij}$ completely. In both settings, the relevant inter-node travel-times are $\tau_{13}=1,\tau_{14}=3,\tau_{23}=2,\tau_{24}=5$ (the rest does not matter for solving $RC(\Lambda)$ for the given demands). These parameters, as well as the support of the corresponding optimal rebalancing flows are indicated by bold arrows in Figures~\ref{fig:4nodeL1} and \ref{fig:4nodeL3}. Note that in the former case, the rebalancing flows comprise a single connected component, and hence the solution is not dual degenerate (via Theorem~\ref{thm:degen}); in contrast, for $\Lambda^2$, optimal rebalancing flow comprises of $2$ connected components, and hence is dual-degenerate.

\begin{figure*}[!ht]
\centering
\subfigure[Fragmentation-Resilient Regime]{
\label{fig:4nodeL1}
    \begin{tikzpicture}[xscale=0.8,yscale=0.8]
    \draw [thick] (0,0) rectangle (1,1);
    \draw [thick] (0,3) rectangle (1,4);
    \draw [thick] (3,0) rectangle (4,1);
    \draw [thick] (3,3) rectangle (4,4);
    \node at (0.5,0.5) {2};
    \node at (0.5,3.5) {1};
    \node at (3.5,0.5) {4};
    \node at (3.5,3.5) {3};
    \draw [->] (1,0.5) -- (3,0.5);
    \draw [->, ultra thick] (1,1) -- (3,3);
    \draw [->, ultra thick] (1,3.5) -- (3,3.5);
    \draw [->,ultra thick] (1,3) -- (3,1);
    \node at (2,0.7) {$\tau_{24}=5$};
    \node at (2,3.2) {$\tau_{13}=1$};
    \node at (0.65,1.6) {$\tau_{23}=2$};
    \node at (0.65,2.4) {$\tau_{14}=3$};
    \draw [->,dashed] (0,0.5) -- (-1,0.5);
    \draw [->,dashed] (0,3.5) -- (-1,3.5);
    \draw [<-,dashed] (4,3.5) -- (5,3.5);
    \draw [<-,dashed] (4,0.5) -- (5,0.5);
    \node at (-0.5,4.2) {$\Lambda_1=2$};
    \node at (-0.5,-0.4) {$\Lambda_2=3$};
    \node at (4.5,4.2) {$\Lambda_3=-4$};
    \node at (4.5,-0.4) {$\Lambda_4=-1$};
    \end{tikzpicture}
}
\hspace{3.5cm}
\subfigure[Fragmentation-Affected Regime]{
\label{fig:4nodeL3}
	\begin{tikzpicture}[xscale=0.8,yscale=0.8]
	\draw [thick] (0,0) rectangle (1,1);
	\draw [thick] (0,3) rectangle (1,4);
	\draw [thick] (3,0) rectangle (4,1);
	\draw [thick] (3,3) rectangle (4,4);
	\node at (0.5,0.5) {2};
	\node at (0.5,3.5) {1};
	\node at (3.5,0.5) {3};
	\node at (3.5,3.5) {4};
	\draw [->] (1,0.5) -- (3,0.5);
	\draw [->, ultra thick] (1,1) -- (3,3);
	\draw [->] (1,3.5) -- (3,3.5);
	\draw [->, ultra thick] (1,3) -- (3,1);
	\node at (2,0.7) {$\tau_{24}=5$};
	\node at (2,3.2) {$\tau_{13}=1$};
	\node at (0.65,1.6) {$\tau_{23}=2$};
	\node at (0.65,2.4) {$\tau_{14}=3$};
    \draw [->,dashed] (0,0.5) -- (-1,0.5);
    \draw [->,dashed] (0,3.5) -- (-1,3.5);
    \draw [<-,dashed] (4,3.5) -- (5,3.5);
    \draw [<-,dashed] (4,0.5) -- (5,0.5);
	\node at (-0.5,4.2) {$\Lambda_1=2$};
	\node at (-0.5,-0.4) {$\Lambda_2=3$};
	\node at (4.5,4.2) {$\Lambda_3=-3$};
	\node at (4.5,-0.4) {$\Lambda_4=-2$};
	\end{tikzpicture}
}
\caption{Four node graph with nodes $\{1,2,3,4\}$ used for synthetic experiments. The dashed arrows indicate the direction of net flow of passengers at a node (i.e., entering or exiting; note $\Lambda_i$ corresponds to total demand flow exiting $i$). Fig.~\ref{fig:4nodeL1} has demand vector $\Lambda^1 = [2,3,-4,-1]$, while Fig.~\ref{fig:4nodeL3} has demand vector $\Lambda^2 = [2,3,-3,-2]$. The relevant edge costs $\tau_{ij}$ (given in the figures) are the same in both settings. The support of the optimal rebalancing flows are indicated by bold arrows in both settings. Note that in the first case, edges with positive rebalancing flow define a single connected component, while in the second, they define two disconnected components $\{1,3\}$ and $\{2,4\}$.}
\label{fig:4node}
\end{figure*}
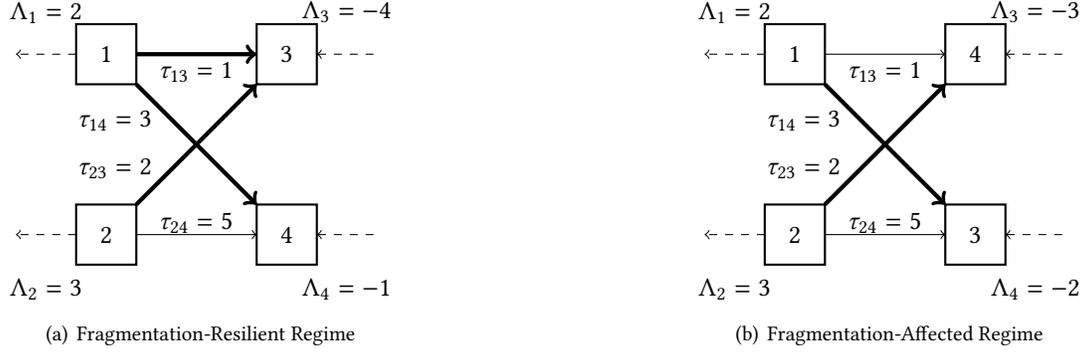
\begin{figure*}[!ht]
\centering
\subfigure[PoF in Fragmentation-Resilient Regime]{
\includegraphics[width=0.45\columnwidth]{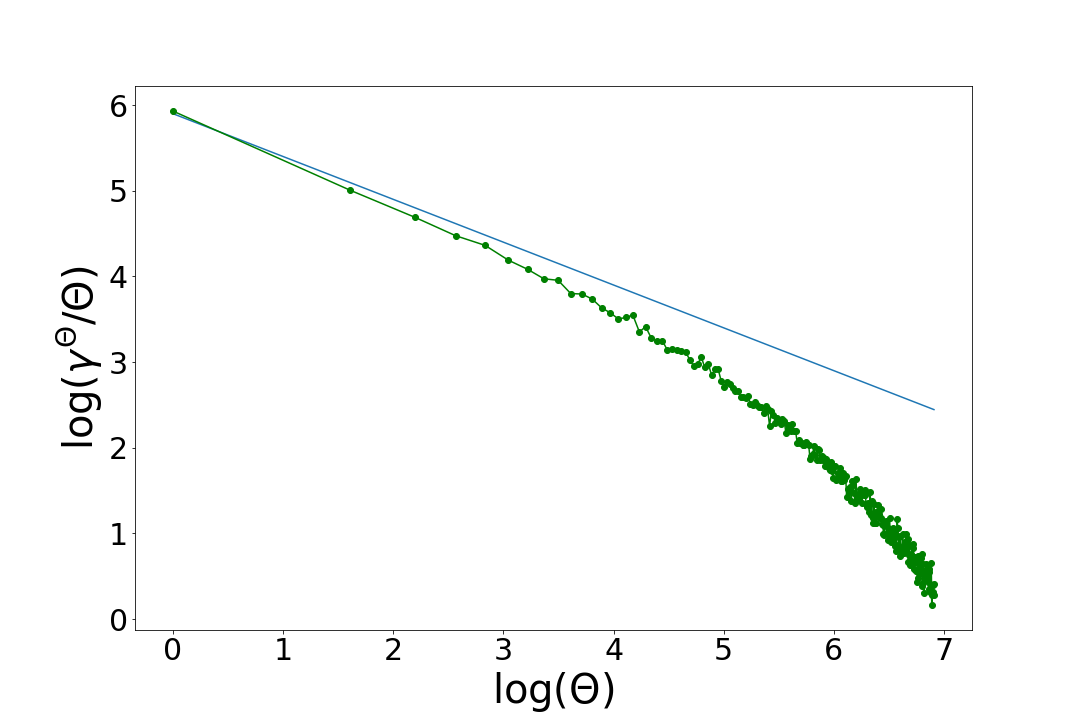}
\label{fig:fourexp}
}
\hspace{0.5cm}
\subfigure[PoF in Fragmentation-Affected Regime]{
\label{fig:foursqrt}
\includegraphics[width=0.45\columnwidth]{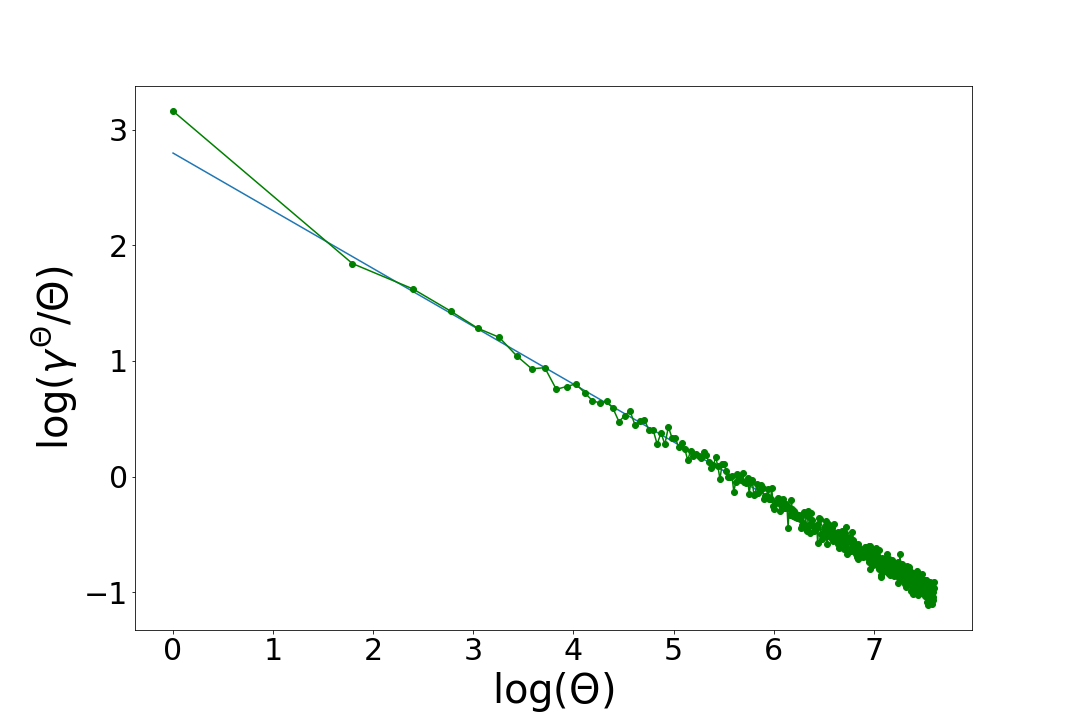}
}
\caption[PoF scaling  in 4-node network]
{Simulation showing Price of Fragmentation in 4-node network in Fig.~\ref{fig:4node}: The above plots show scaling behavior of the (normalized) increase in rebalancing costs $\lambda^{\theta}/\theta$ in the settings described above. In particular Fig.~\ref{fig:fourexp} corresponds to the fragmentation-resilient setting (Fig.~\ref{fig:4nodeL1}), while Fig.~\ref{fig:foursqrt} corresponds to the fragmentation-affected setting (Fig.~\ref{fig:4nodeL3}).}
\label{fig:preview}
\end{figure*}

In each case, to compute $\gamma^\theta/\theta$ for each $\theta$, we simulate $500$ instances of our random splitting process, and aggregate the random rebalancing costs. 

Fig.~\ref{fig:fourexp} and \ref{fig:foursqrt} show the scaling behavior of the PoF (in terms of $\log(\gamma^{\theta}/\theta)$ vs. $\log\theta$ in the two settings.
The green curve plots PoF estimates obtained from simulations, while the blue curve is a reference line of slope $-0.5$. Note that in the dual-degenerate demand case, the asymptotic regime indeed shows square-root growth, while in the other regime, the decay is sub-polynomial.

In addition, we also compare PoF under Binomial and Poisson original demands (and random splitting) under the fragmentation-resilient setting in Fig.~\ref{fig:BinPoiComp}. Here we see that the latter decays slower than the the former, as suggested by our theoretical results in Theorem~\ref{thm:expdecay}. 

\begin{figure}[H]
\begin{center}
\includegraphics[width=0.5\columnwidth]{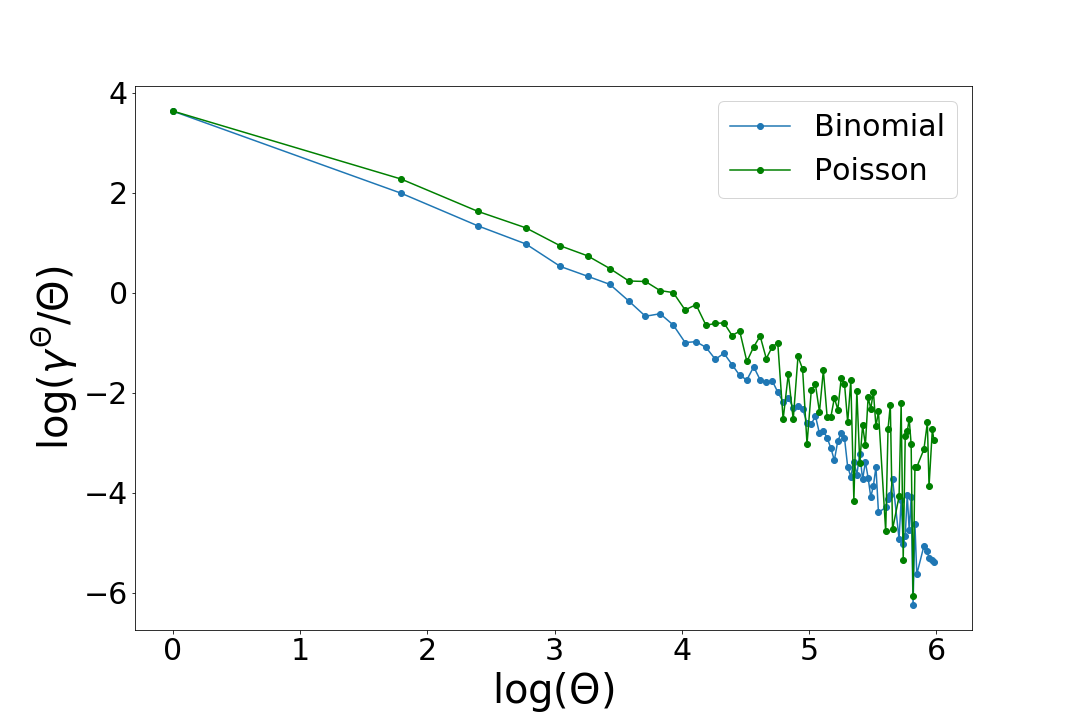}
\caption{Comparison of Binomial and Poisson demand splitting in the fragmentation-affected regime (Fig.~\ref{fig:fourexp}).}
\label{fig:BinPoiComp}
\end{center}
\end{figure}

\subsection{Experiments on NYC Dataset}
\label{ssec:nycexpts}

While we are able to experimentally demonstrate the two regimes in synthetic networks, the question remains as to whether both regimes are only witnessed in carefully engineered examples (in particular, since fragmentation-affected regimes occur under dual-degenerate demands, which may be unlikely in practice). 
To answer this question, we extend our experimental analysis to a simplified network model of New York City, with demand data from the New York City Taxi Dataset~\cite{donovan2014new} as a proxy for the actual system demand.\\ 
\noindent\textbf{Experimental setup}: Our experiments use taxi data recorded for the months of May and June 2016. The dataset has about $335,000$ trips per day; on an hourly basis, the number of trips ranges from $3,600$ to $20,000$ with an average number of around $14,000$ thousands trips per hour. 
We model the city as consisting of a finite number of pickup and dropoff \textit{stations} (20-80 in our experiments), where we assume each trip begins and ends. The stations are determined by taking all of the user demand for the full period of study, clustering the data (pickup and dropoff locations) into the desired number of stations (clusters), and picking the centroid for each cluster. For each time window of interest, the demands are then aggregated to the total demand per origin-destination station pair $\Lambda_{ij}$. For convenience, we use the distance $d_{ij}$ as a proxy for the travel-times $\tau_{ij}$ to parametrize edge-costs; we compute the Manhattan distance between the stations to get the distances $d_{ij}$.  The pickup-dropoff station locations are shown in Fig.~\ref{fig:cluster} in Appendix~\ref{appsec:plots}.

\begin{figure*}[!ht]
	\centering
	\subfigure[PoF on linear scale]{
		\includegraphics[width=0.45\columnwidth]{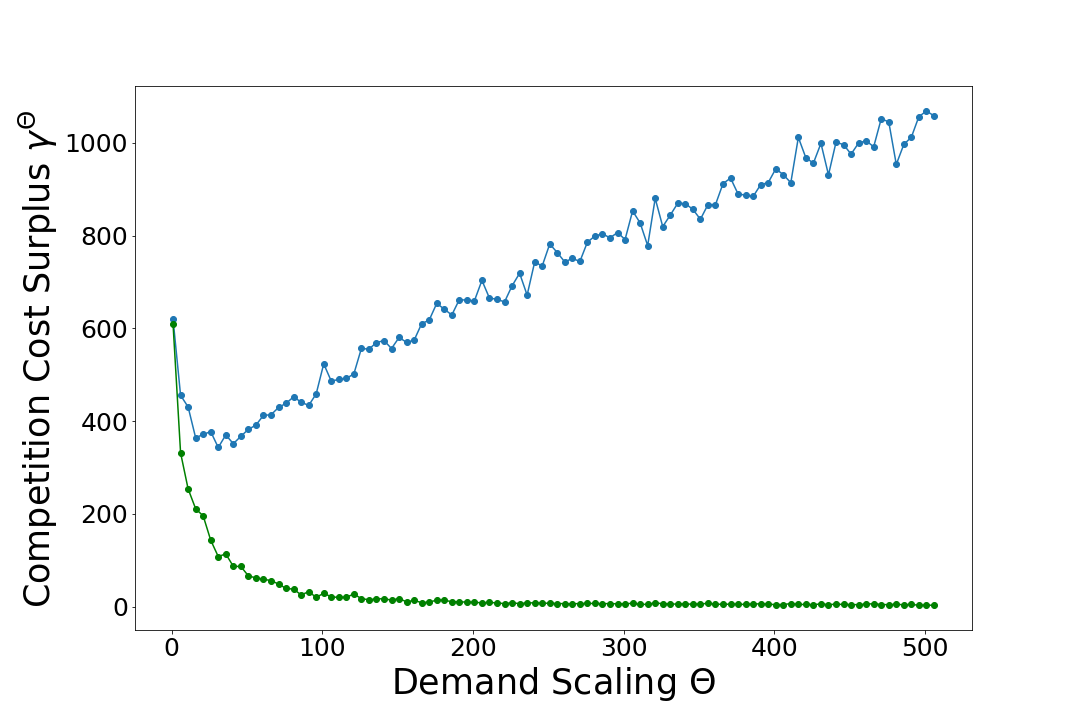}
		\label{fig:intro-full}
	}
	\hspace{1cm}
	\subfigure[PoF on log-log scale]{
		\label{fig:NYCregime}
		\includegraphics[width=0.45\columnwidth]{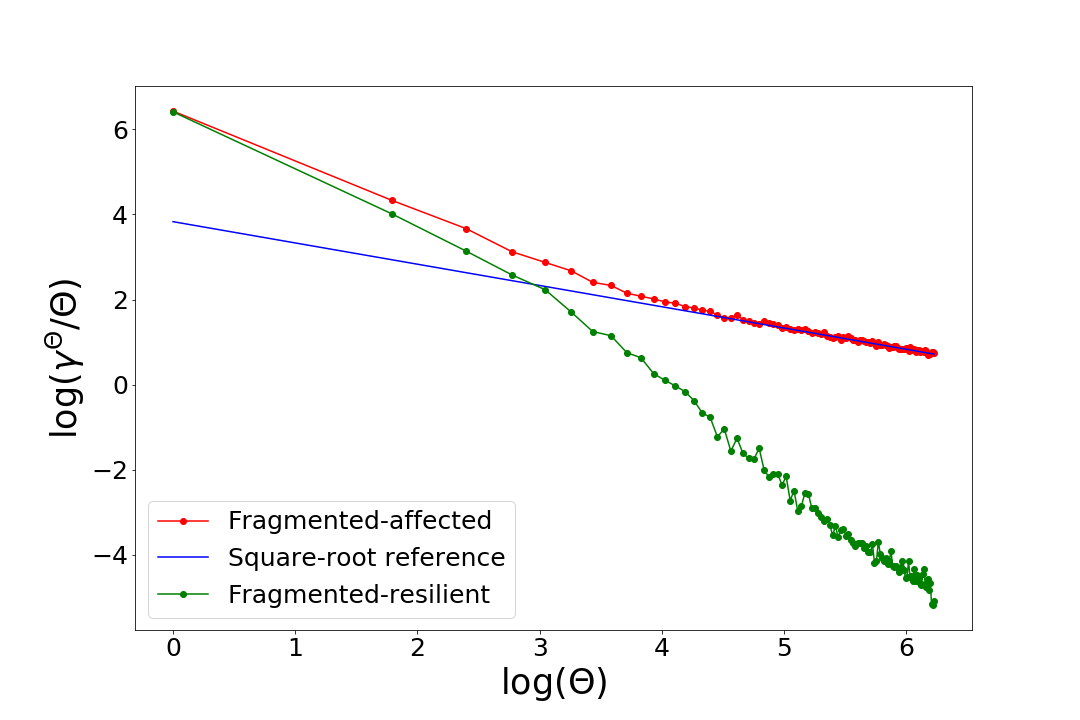}
	}
	\caption[PoF scaling  with NYC data]
	{Detailed plots for simulated data shown in Fig.~\ref{fig:intro}. Fig.~\ref{fig:intro-full} shows PoF scaling for a wider range of $\theta$, allowing us to observe the asymptotic behaviour; here the blue line corresponds to the fragmentation-affected regime, and the green to the fragmentation-resilient regime. Fig.~\ref{fig:NYCregime} displays the same data in log-log scale, allowing us to observe precisely the scaling.}
	\label{fig:NYCsimul}
\end{figure*}

\begin{figure*}[!th]
\centering
\subfigure[Effect of number stations on affected regime probability]{
\includegraphics[width=0.45\columnwidth]{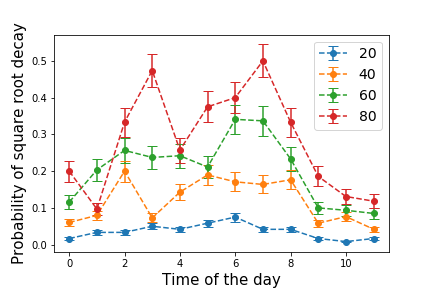}
\label{fig:NYCprob1}
}
\hspace{1cm}
\subfigure[Effect of time window on affected regime probability]{
\label{fig:changeTw}
\includegraphics[width=0.45\columnwidth]{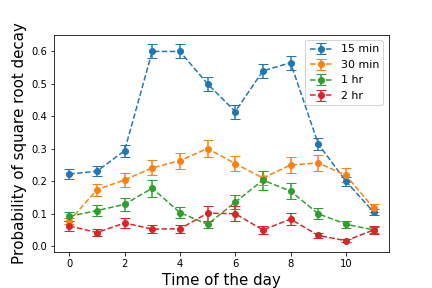}
}
\caption[Probability of affected regime]
{Probability of observing fragmentation-affected demand in NYC based on the NYC taxi dataset. Fig.~\ref{fig:NYCprob1} shows the influence of the number of stations, while Fig.~\ref{fig:changeTw} demonstrates the influence of the length of the time windows. Statistics plotted are based on aggregation of the data into time windows of length 2 hours.}
\label{fig:NYCproba}
\end{figure*}

\noindent\textbf{Experimental Findings}: In Figs.~\ref{fig:NYCregime} and ~\ref{fig:intro-full}, we show the simulation data from Fig.~\ref{fig:intro} in more detail; this provides a clear indication that both PoF regimes can be observed using demand distributions estimated from the NYC data.
In particular, the log-log plot in Fig.~\ref{fig:intro-full} clearly demonstrates the square-root growth and exponential decay. Note also that the two instances are very similar in terms of time (12pm vs. 1pm) and PoF for $\theta=1$; the difference really manifests under scaling.

Since we have established that the two regimes only depend on the structure of the underlying monopolist demand flows (and is testable via Theorem~\ref{thm:degen}), this allows us to quantify how often the fragmented-affected regime occurs in the NYC data. 
One issue in performing such an analysis is the choice of different estimation parameters, in particular, number of stations and size of aggregation time window. To this end, in Figs.~\ref{fig:NYCprob1} and \ref{fig:changeTw}, we plot the probability of observing the fragmentation-affected regime (along with the 95 percentile error bounds) for different time windows, and varying number of stations, using the data for the months of May and June 2016. We note that we compute the probability of dual-degeneracy conditional on having a connected demand graph (this turns out to be the case in almost all settings).

In Fig.~\ref{fig:NYCprob1}, we observe that the probability of observing dual-degenerate demands increases significantly as we increase the number of stations. This can be partially explained by the following heuristic argument: in a system with $N$ stations, the demand is scattered over $N^2$ edges. Consequently, the sparsity of the demand matrix increases quadratically, and when solving the LP, fewer edges are active, leading to more disconnected components (and therefore more fragmented-affected regimes, using Theorem~\ref{thm:degen}). Increasing the time window to aggregate more demands reduces the likelihood of this phenomenon, but does not eliminate it. In practice, ride-hailing systems currently pickup and drop-off passengers at a spatial granularity corresponding to each street corner, which is approximately 4000 stations for NYC; based on the trend in Fig.~\ref{fig:NYCprob1}, this may imply a very high probability of observing the fragmented affected regime.

One concern regarding the above discussion is that increasing number of stations also decreases the average inter-station travel cost; hence the overall effect on PoF is unclear. To this end, we also directly computed the PoF scaling behaviour of the system as a function of the number of stations, using a fixed demand (again corresponding to 1-2pm on the $10^{th}$ of May 2016). In Fig.~\ref{fig:changeStation}, we plot this scaling corresponding to $[20,40,60,80]$ stations.  
We observe that the setting appears to be fragmentation-resilient when aggregated using 20 and 40 stations, but fragmentation-affected for higher resolutions. 
We note though that this trend does not hold for all time slots, and other trends (e.g., being fragmentation-resilient with 20 and 60 stations, but not the others) can also be observed in the data; however, the overall trend (in Fig.\ref{fig:NYCprob1}) seems to suggest that greater resolutions lead to greater probability of observing fragmentation-affected regimes.

(We note also that the initial PoF, i.e., for $\theta=1$, increases as the spatial granularity increases. This is a numerical artifact caused by demands whose pick-up and drop-off locations are in the same cluster decreasing as the spatial granularity is refined -- recall that a station based model assumes no loss for demands that start and end at the same station.)

Another variable that impacts the aggregation of data, and thus the likelihood of degeneracy, is the length of the time window, as shown in Fig.~\ref{fig:changeTw} (with number of stations fixed to 40). One interpretation for this is that it is caused due to a discretization phenomenon -- shorter time windows may lead to sparser entries in the rebalancing LP, due to which there may be a higher probability of obtaining a dual-degenerate rebalancing flow. In any case, since the overall trend (higher probability of fragmentation-affected regimes when moving to smaller time slots) suggests that there indeed is a non-trivial probability of observing fragmentation-affectedness in the NYC data, no matter how we slice it.

\begin{figure}[!t]
\begin{center}
\includegraphics[width=0.5\columnwidth]{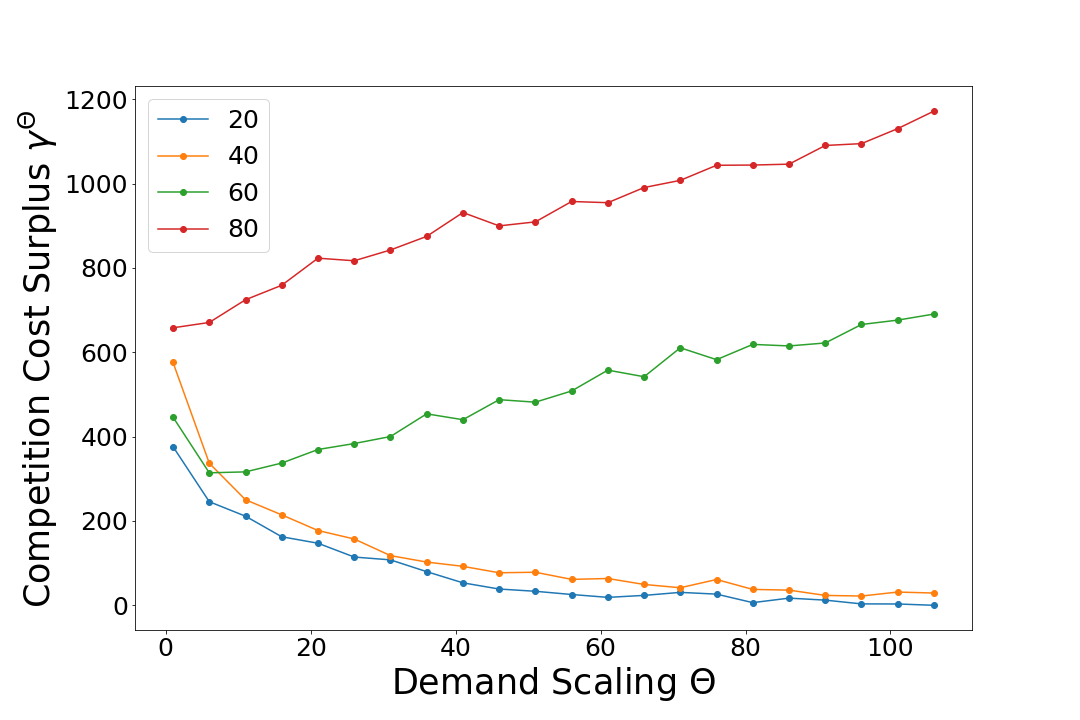}
\caption{Influence of the number of stations on PoF scaling (based on 1pm-2pm data for May 10).}
\label{fig:changeStation}
\end{center}
\end{figure}

Finally, we explore the effect of adding more firms to the ecosystem, as well as the effect of having inhomogeneous market shares; these plots are provided in Appendix~\ref{appsec:plots}. Adding more firms leads to the similar behaviour as with two firms (this confirms what we show in Theorem~\ref{thm:anycompany}). Inhomogeneous market shares however cause a significant increase in the PoF, and as predicted in Theorem~\ref{appthm:nonuniform}, exhibit a linear scaling (See Fig.~\ref{fig:inhomogen}). Furthermore the linear divergence gives PoF 10 to 100 times bigger compared to homogeneous shares when you scale. This shows that homogeneizing demand might be a key aspect for policy makers to try and control, if the firms indeed have spatially varying market shares.

\section{Insights and Policy Recommendations}
\label{sec:insights}

We now briefly discuss the insights we obtain from our theoretical and simulation results, and how they may be used to inform and shape public policy for regulating ride-hailing ecosystems.

\subsection{Insights}



At a high level, our results suggest that fragmentation in an MoD ecosystem incurs minimal additional cost as long as the underlying demands are not degenerate with respect to the dual of the minimum cost rebalancing LP. When the demand matrix is degenerate, however, we get the fragmentation-affected regime, wherein we have PoF$(\theta)=\Theta(\theta^{1/2})$; here, even though there is zero normalized loss between the duopoly and the monopoly in the limit, the fluctuations are such that the actual loss diverges when scaling $\theta$. Theorem~\ref{thm:degen} gives us a testable way of checking for this setting -- moreover, it provides an intuitive condition for this in terms of the existence of {\em locally balanced clusters} -- these correspond to subsets of the network wherein all nodes are relatively close by, and which have large balanced flows across the cut~\footnote{For example, consider the flow of traffic between San Francisco and Berkeley across the Bay Bridge -- there is a significant cost in traversing between the two, and moreover, a roughly balanced flow of people traveling in both directions.}. 

Numerically, we observe that this degeneracy condition is quite sensitive to the spatio-temporal discretization of the demand, i.e. the number of stations and time window. In particular, the probability of degeneracy appears to increase with the spatio-temporal granularity of the model. 
One important insight from this is that {\em aggregating pickup and drop off locations, and increasing the flexibility of pickup times} can have a large impact on the performance of a fragmented ride-hailing system. Practically speaking, cities do have policy/regulatory tools at their disposal that could help in this direction. For example, given that cities control curbside access, they can designate dedicated pickup and drop off locations (stations) that are spaced out in a manner that decreases the probability of degeneracy. Ride-hailing companies can also provide incentives to encourage customers to walk to stations that help with the rebalancing problem. On the temporal aggregation side, incentives can be provided to passengers who have flexibility in their pickup times. 

Another important insight of our results is regarding the importance of {\em multi-homing}, i.e., drivers who work for multiple firms simultaneously. This creates flexibility on the supply side -- in particular, it is easy to see that significant efficiencies can be obtained if some subset of the drivers can be summoned to serve demands on either platform (company). To see this, note that the average fluctuations in our model are on the order of $\sqrt{\theta}$; now if a constant fraction of drivers are multi-homing, there there is always sufficient flexible supply available to prevent the need for rebalancing. Multi-homing is a phenomenon that already occurs with many ride-hailing drivers in the US simultaneously being active on both the Lyft and Uber apps. The situation can be further improved by cities mandating that drivers (who are independent contractors) must be allowed to operate across multiple platforms and creating more awareness about the possibility of doing so. This is not allowed by ride-hailing companies in certain cities. Moroever, with the advent of driverless fleets, our work suggests that it may be desirable to require that MoD platforms continue to have a significant number of multi-homing drivers in their fleets.

Another insight from our analysis is that inhomogeneous market shares may lead to a greater chance of being in a fragmentation-affected regime, and moreover, can lead to much worse PoF (see Fig.\ref{fig:inhomogen}). Thus, operational settings where different firms target different neighbourhoods can lead to an increased probability of degeneracy. Unfortunately, this is a difficult problem to tackle via a policy/regulatory approach. One thought is that cities can impose a geographic equity requirement on providers that could help with this. An alternate more involved options is to set up a centralized market maker who controls a fraction of the flow (for example, a third-party app run by the city or local transit agency), and can use this to make demand-splits homogeneous over the city. Moreover, the market maker can directly target its actions to eliminate the degeneracy condition, if it has complete knowledge about all the service providers in the market, admittedly a challenging if not impossible task. 

Finally, we note that the conditions we obtain for fragmentation-affectedness are in a sense the very conditions under which \emph{mass public transit is efficient} -- people wanting to travel between some two distant locations, in large numbers, but moreover, in equal numbers in both directions. Thus, in a sense, by suggesting that fragmentation-affected settings can be remedied in some cases by having a better public transit infrastructure, our results provide more support for having good public transit.

\subsection{Extensions}
Our work presents a simple stylistic model for formally thinking about the PoF in ride-hailing systems. We now discuss some ideas for extending this initial work to more complex settings. One way to add more flexibility to the model is to extend it via a complimentary data-driven simulation model. Using our model as a starting point, and then developing more refined simulation based on the NYC taxi dataset. This would allow us to study the PoF assuming different firms use dynamic fleet management policies that react to demand surges. 

Next, our work focused on one metric, the cost of rebalancing, to quantify the societal impact of MoD systems. This has the advantage of being simple -- however, it raises the question of how our results change when we consider other controls which increase flexibility, in particular, reservation mechanisms and ride-pooling. 

Although our work studies the effect of fragmentation, we do so in a regime where we ignore strategic dynamics between the platforms. Knowing that there is a regime where the consequences of fragmentation do not vanish when scaling, a natural question is if we can exploit these competitive dynamics to affect a change in the PoF. As an extreme example, we can imagine a system where all demands are pooled to constitute a central server, that then allocates requests across companies. This relates to ideas of limited resource pooling in queueing networks~\cite{TJK12}, which are known to have dramatic impacts on the performance. 

Finally, any new policy that reduces rebalancing costs is meaningful only if it is economically viable. In essence, what we want is a system that can coordinate global demand and global fleet, but also sustains a competitive market. This can be modeled as a Stackelberg game framework, and ideas from existing models that incorporate queuing dynamics~\cite{hassin2003queue,lingenbrink2017optimal} may prove useful in our setting.

\section{Conclusions}
\label{sec:conclusions}

While competition between multiple service providers leads to better outcomes for consumers in general, the work presented here shows that demand fragmentation across multiple platforms in MoD markets can lead to significant losses in terms of operating costs. To this end, we show (both theoretically and experimentally) that the system admits a phase transition in terms of the {\em Price of Fragmentation} (PoF), a metric we define which quantifies the increase in operational costs in a duopoly vs. a monopoly under a stochastic demand splitting model. More precisely, under any demand-splitting process, we show that the system will admit two dramatically different behaviors when the demand is scaled up to infinity; one with an asymptotic loss of zero, and another with an infinite asymptotic loss. These results on the phase transition remain valid even in situations where the market shares are not homogeneous over the city and there are more than two firms competing.

This article is an initial treatment of a subject that is of considerable interest with the increasing importance of MoD systems and their potential for being part of a sustainable transportation ecosystem for urban areas in the future. Furthermore, the losses that we observe should be even more pronounced in MoD systems with ridepooling, where multiple demands are served by the same vehicle and demand aggregation is critical for efficiency. Given that merging all service operators to create a monopoly is not a viable solution (for many reasons), there is a need for developing new mechanisms that can mitigate this efficiency loss. This work aims to motivate future research in that direction.


\bibliographystyle{ACM-Reference-Format}
\bibliography{bibliography} 


\begin{thebibliography}{24}


\ifx \showCODEN    \undefined \def \showCODEN     #1{\unskip}     \fi
\ifx \showDOI      \undefined \def \showDOI       #1{#1}\fi
\ifx \showISBNx    \undefined \def \showISBNx     #1{\unskip}     \fi
\ifx \showISBNxiii \undefined \def \showISBNxiii  #1{\unskip}     \fi
\ifx \showISSN     \undefined \def \showISSN      #1{\unskip}     \fi
\ifx \showLCCN     \undefined \def \showLCCN      #1{\unskip}     \fi
\ifx \shownote     \undefined \def \shownote      #1{#1}          \fi
\ifx \showarticletitle \undefined \def \showarticletitle #1{#1}   \fi
\ifx \showURL      \undefined \def \showURL       {\relax}        \fi
\providecommand\bibfield[2]{#2}
\providecommand\bibinfo[2]{#2}
\providecommand\natexlab[1]{#1}
\providecommand\showeprint[2][]{arXiv:#2}

\bibitem[\protect\citeauthoryear{??}{don}{2017}]%
        {donovan2014new}
 \bibinfo{year}{2017}\natexlab{}.
\newblock \bibinfo{title}{NYC Taxi Data}.
\newblock   (\bibinfo{year}{2017}).
\newblock
\showURL{%
\url{http://www.nyc.gov/html/tlc/html/about/trip_record_data.shtml}}


\bibitem[\protect\citeauthoryear{Abolhassani, Bateni, Hajiaghayi, Mahini, and
  Sawant}{Abolhassani et~al\mbox{.}}{2014}]%
        {ABH14}
\bibfield{author}{\bibinfo{person}{Melika Abolhassani},
  \bibinfo{person}{Mohammad~Hossein Bateni}, \bibinfo{person}{MohammadTaghi
  Hajiaghayi}, \bibinfo{person}{Hamid Mahini}, {and} \bibinfo{person}{Anshul
  Sawant}.} \bibinfo{year}{2014}\natexlab{}.
\newblock \showarticletitle{Network cournot competition}. In
  \bibinfo{booktitle}{{\em International Conference on Web and Internet
  Economics}}. Springer, \bibinfo{pages}{15--29}.
\newblock


\bibitem[\protect\citeauthoryear{Alonso-Mora, Samaranayake, Wallar, Frazzoli,
  and Rus}{Alonso-Mora et~al\mbox{.}}{2017}]%
        {SFR17}
\bibfield{author}{\bibinfo{person}{Javier Alonso-Mora},
  \bibinfo{person}{Samitha Samaranayake}, \bibinfo{person}{Alex Wallar},
  \bibinfo{person}{Emilio Frazzoli}, {and} \bibinfo{person}{Daniela Rus}.}
  \bibinfo{year}{2017}\natexlab{}.
\newblock \showarticletitle{On-demand high-capacity ride-sharing via dynamic
  trip-vehicle assignment}.
\newblock \bibinfo{journal}{{\em Proceedings of the National Academy of
  Sciences\/}} (\bibinfo{year}{2017}), \bibinfo{pages}{201611675}.
\newblock


\bibitem[\protect\citeauthoryear{Anshelevich and Sekar}{Anshelevich and
  Sekar}{2015}]%
        {AS15}
\bibfield{author}{\bibinfo{person}{Elliot Anshelevich} {and}
  \bibinfo{person}{Shreyas Sekar}.} \bibinfo{year}{2015}\natexlab{}.
\newblock \showarticletitle{Price Competition in Networked Markets: How do
  monopolies impact social welfare?}. In \bibinfo{booktitle}{{\em International
  Conference on Web and Internet Economics}}. Springer,
  \bibinfo{pages}{16--30}.
\newblock


\bibitem[\protect\citeauthoryear{Banerjee, Freund, and Lykouris}{Banerjee
  et~al\mbox{.}}{2017}]%
        {BFL17}
\bibfield{author}{\bibinfo{person}{Siddhartha Banerjee},
  \bibinfo{person}{Daniel Freund}, {and} \bibinfo{person}{Thodoris Lykouris}.}
  \bibinfo{year}{2017}\natexlab{}.
\newblock \showarticletitle{Pricing and Optimization in Shared Vehicle Systems:
  An Approximation Framework}. In \bibinfo{booktitle}{{\em Proceedings of the
  2017 ACM Conference on Economics and Computation}}. ACM,
  \bibinfo{pages}{517--517}.
\newblock


\bibitem[\protect\citeauthoryear{Banerjee, Johari, and Riquelme}{Banerjee
  et~al\mbox{.}}{2015}]%
        {banerjee2015pricing}
\bibfield{author}{\bibinfo{person}{Siddhartha Banerjee},
  \bibinfo{person}{Ramesh Johari}, {and} \bibinfo{person}{Carlos Riquelme}.}
  \bibinfo{year}{2015}\natexlab{}.
\newblock \showarticletitle{Pricing in ride-sharing platforms: A
  queueing-theoretic approach}. In \bibinfo{booktitle}{{\em Proceedings of the
  Sixteenth ACM Conference on Economics and Computation}}. ACM,
  \bibinfo{pages}{639--639}.
\newblock


\bibitem[\protect\citeauthoryear{Bimpikis, Ehsani, and Ilkilic}{Bimpikis
  et~al\mbox{.}}{2014}]%
        {BEI14}
\bibfield{author}{\bibinfo{person}{Kostas Bimpikis}, \bibinfo{person}{Shayan
  Ehsani}, {and} \bibinfo{person}{Rahmi Ilkilic}.}
  \bibinfo{year}{2014}\natexlab{}.
\newblock \showarticletitle{Cournot competition in networked markets.}. In
  \bibinfo{booktitle}{{\em EC}}. \bibinfo{pages}{733}.
\newblock


\bibitem[\protect\citeauthoryear{Braverman, Dai, Liu, and Ying}{Braverman
  et~al\mbox{.}}{2016}]%
        {braverman2016empty}
\bibfield{author}{\bibinfo{person}{Anton Braverman}, \bibinfo{person}{JG Dai},
  \bibinfo{person}{Xin Liu}, {and} \bibinfo{person}{Lei Ying}.}
  \bibinfo{year}{2016}\natexlab{}.
\newblock \showarticletitle{Empty-car routing in ridesharing systems}.
\newblock \bibinfo{journal}{{\em arXiv preprint arXiv:1609.07219\/}}
  (\bibinfo{year}{2016}).
\newblock


\bibitem[\protect\citeauthoryear{Cai, Bose, and Wierman}{Cai
  et~al\mbox{.}}{2017}]%
        {CBW17}
\bibfield{author}{\bibinfo{person}{Desmond Cai}, \bibinfo{person}{Subhonmesh
  Bose}, {and} \bibinfo{person}{Adam Wierman}.}
  \bibinfo{year}{2017}\natexlab{}.
\newblock \showarticletitle{On the role of a market maker in networked cournot
  competition}.
\newblock \bibinfo{journal}{{\em arXiv preprint arXiv:1701.08896\/}}
  (\bibinfo{year}{2017}).
\newblock


\bibitem[\protect\citeauthoryear{George}{George}{2012}]%
        {george2012stochastic}
\bibfield{author}{\bibinfo{person}{David~K George}.}
  \bibinfo{year}{2012}\natexlab{}.
\newblock {\em \bibinfo{title}{Stochastic Modeling and Decentralized Control
  Policies for Large-Scale Vehicle Sharing Systems via Closed Queueing
  Networks}}.
\newblock \bibinfo{thesistype}{Ph.D. Dissertation}. \bibinfo{school}{The Ohio
  State University}.
\newblock


\bibitem[\protect\citeauthoryear{Hassin and Haviv}{Hassin and Haviv}{2003}]%
        {hassin2003queue}
\bibfield{author}{\bibinfo{person}{Refael Hassin} {and} \bibinfo{person}{Moshe
  Haviv}.} \bibinfo{year}{2003}\natexlab{}.
\newblock \bibinfo{booktitle}{{\em To queue or not to queue: Equilibrium
  behavior in queueing systems}}. Vol.~\bibinfo{volume}{59}.
\newblock \bibinfo{publisher}{Springer Science \& Business Media}.
\newblock


\bibitem[\protect\citeauthoryear{Lingenbrink and Iyer}{Lingenbrink and
  Iyer}{2017}]%
        {lingenbrink2017optimal}
\bibfield{author}{\bibinfo{person}{David Lingenbrink} {and}
  \bibinfo{person}{Krishnamurthy Iyer}.} \bibinfo{year}{2017}\natexlab{}.
\newblock \showarticletitle{Optimal Signaling Mechanisms in Unobservable Queues
  with Strategic Customers}.
\newblock  (\bibinfo{year}{2017}).
\newblock


\bibitem[\protect\citeauthoryear{of~Economic and Affairs}{of~Economic and
  Affairs}{2015}]%
        {UN14}
\bibfield{author}{\bibinfo{person}{United Nations~Department of Economic} {and}
  \bibinfo{person}{Social Affairs}.} \bibinfo{year}{2015}\natexlab{}.
\newblock \bibinfo{booktitle}{{\em World Urbanization Prospects: The 2014
  Revision}}.
\newblock \bibinfo{type}{{T}echnical {R}eport}. \bibinfo{institution}{United
  Nations}.
\newblock


\bibitem[\protect\citeauthoryear{Ozkan and Ward}{Ozkan and Ward}{2016}]%
        {ozkan2016dynamic}
\bibfield{author}{\bibinfo{person}{Erhun Ozkan} {and} \bibinfo{person}{Amy~R
  Ward}.} \bibinfo{year}{2016}\natexlab{}.
\newblock \showarticletitle{Dynamic Matching for Real-Time Ridesharing}.
\newblock  (\bibinfo{year}{2016}).
\newblock


\bibitem[\protect\citeauthoryear{Pavone, Smith, Frazzoli, and Rus}{Pavone
  et~al\mbox{.}}{2012}]%
        {PSFR12}
\bibfield{author}{\bibinfo{person}{Marco Pavone}, \bibinfo{person}{Stephen~L
  Smith}, \bibinfo{person}{Emilio Frazzoli}, {and} \bibinfo{person}{Daniela
  Rus}.} \bibinfo{year}{2012}\natexlab{}.
\newblock \showarticletitle{Robotic load balancing for mobility-on-demand
  systems}.
\newblock \bibinfo{journal}{{\em The International Journal of Robotics
  Research\/}} \bibinfo{volume}{31}, \bibinfo{number}{7}
  (\bibinfo{year}{2012}), \bibinfo{pages}{839--854}.
\newblock


\bibitem[\protect\citeauthoryear{Rochet and Tirole}{Rochet and Tirole}{2006}]%
        {rochet2006two}
\bibfield{author}{\bibinfo{person}{Jean-Charles Rochet} {and}
  \bibinfo{person}{Jean Tirole}.} \bibinfo{year}{2006}\natexlab{}.
\newblock \showarticletitle{Two-sided markets: a progress report}.
\newblock \bibinfo{journal}{{\em The RAND Journal of Economics\/}}
  \bibinfo{volume}{37}, \bibinfo{number}{3} (\bibinfo{year}{2006}),
  \bibinfo{pages}{645--667}.
\newblock


\bibitem[\protect\citeauthoryear{Rysman}{Rysman}{2009}]%
        {rysman2009economics}
\bibfield{author}{\bibinfo{person}{Marc Rysman}.}
  \bibinfo{year}{2009}\natexlab{}.
\newblock \showarticletitle{The economics of two-sided markets}.
\newblock \bibinfo{journal}{{\em The Journal of Economic Perspectives\/}}
  (\bibinfo{year}{2009}), \bibinfo{pages}{125--143}.
\newblock


\bibitem[\protect\citeauthoryear{Santi, Resta, Szell, Sobolevsky, Strogatz, and
  Ratti}{Santi et~al\mbox{.}}{2014}]%
        {SR14}
\bibfield{author}{\bibinfo{person}{Paolo Santi}, \bibinfo{person}{Giovanni
  Resta}, \bibinfo{person}{Michael Szell}, \bibinfo{person}{Stanislav
  Sobolevsky}, \bibinfo{person}{Steven~H Strogatz}, {and}
  \bibinfo{person}{Carlo Ratti}.} \bibinfo{year}{2014}\natexlab{}.
\newblock \showarticletitle{Quantifying the benefits of vehicle pooling with
  shareability networks}.
\newblock \bibinfo{journal}{{\em Proceedings of the National Academy of
  Sciences\/}} \bibinfo{volume}{111}, \bibinfo{number}{37}
  (\bibinfo{year}{2014}), \bibinfo{pages}{13290--13294}.
\newblock


\bibitem[\protect\citeauthoryear{Schrank, Schrank, Lomax, and Bak}{Schrank
  et~al\mbox{.}}{2015}]%
        {TTI15}
\bibfield{author}{\bibinfo{person}{David Schrank}, \bibinfo{person}{David
  Schrank}, \bibinfo{person}{Tim Lomax}, {and} \bibinfo{person}{Jim Bak}.}
  \bibinfo{year}{2015}\natexlab{}.
\newblock \bibinfo{booktitle}{{\em 2015 Urban Mobility Scorecard}}.
\newblock \bibinfo{type}{{T}echnical {R}eport}. \bibinfo{institution}{The Texas
  A and M Transportation Institute and INRIX}.
\newblock


\bibitem[\protect\citeauthoryear{Spieser, Samaranayake, Gruel, and
  Frazzoli}{Spieser et~al\mbox{.}}{2016}]%
        {SGF15}
\bibfield{author}{\bibinfo{person}{Kevin Spieser}, \bibinfo{person}{Samitha
  Samaranayake}, \bibinfo{person}{Wolfgang Gruel}, {and}
  \bibinfo{person}{Emilio Frazzoli}.} \bibinfo{year}{2016}\natexlab{}.
\newblock \showarticletitle{Shared-vehicle mobility-on-demand systems: a fleet
  operator's guide to rebalancing empty vehicles}. In \bibinfo{booktitle}{{\em
  Transportation Research Board 95th Annual Meeting}}.
\newblock


\bibitem[\protect\citeauthoryear{Treleaven, Pavone, and Frazzoli}{Treleaven
  et~al\mbox{.}}{2013}]%
        {TPF13}
\bibfield{author}{\bibinfo{person}{Kyle Treleaven}, \bibinfo{person}{Marco
  Pavone}, {and} \bibinfo{person}{Emilio Frazzoli}.}
  \bibinfo{year}{2013}\natexlab{}.
\newblock \showarticletitle{Asymptotically optimal algorithms for one-to-one
  pickup and delivery problems with applications to transportation systems}.
\newblock \bibinfo{journal}{{\it IEEE Trans. Automat. Control}}
  \bibinfo{volume}{58}, \bibinfo{number}{9} (\bibinfo{year}{2013}),
  \bibinfo{pages}{2261--2276}.
\newblock


\bibitem[\protect\citeauthoryear{Tsitsiklis, Xu, et~al\mbox{.}}{Tsitsiklis
  et~al\mbox{.}}{2012}]%
        {TJK12}
\bibfield{author}{\bibinfo{person}{John~N Tsitsiklis}, \bibinfo{person}{Kuang
  Xu}, {et~al\mbox{.}}} \bibinfo{year}{2012}\natexlab{}.
\newblock \showarticletitle{On the power of (even a little) resource pooling}.
\newblock \bibinfo{journal}{{\em Stochastic Systems\/}} \bibinfo{volume}{2},
  \bibinfo{number}{1} (\bibinfo{year}{2012}), \bibinfo{pages}{1--66}.
\newblock


\bibitem[\protect\citeauthoryear{Weyl}{Weyl}{2010}]%
        {weyl2010price}
\bibfield{author}{\bibinfo{person}{E~Glen Weyl}.}
  \bibinfo{year}{2010}\natexlab{}.
\newblock \showarticletitle{A price theory of multi-sided platforms}.
\newblock \bibinfo{journal}{{\em The American Economic Review\/}}
  (\bibinfo{year}{2010}), \bibinfo{pages}{1642--1672}.
\newblock


\bibitem[\protect\citeauthoryear{Zhang and Pavone}{Zhang and Pavone}{2016}]%
        {ZP14}
\bibfield{author}{\bibinfo{person}{Rick Zhang} {and} \bibinfo{person}{Marco
  Pavone}.} \bibinfo{year}{2016}\natexlab{}.
\newblock \showarticletitle{Control of robotic mobility-on-demand systems: a
  queueing-theoretical perspective}.
\newblock \bibinfo{journal}{{\em The International Journal of Robotics
  Research\/}} \bibinfo{volume}{35}, \bibinfo{number}{1-3}
  (\bibinfo{year}{2016}), \bibinfo{pages}{186--203}.
\newblock


\end{thebibliography}

\newpage
\appendix

\section{Proofs of Results}
\label{appsec:proofs}

In this appendix, we provide the complete proofs of our results. For convenience, we first restate the relevant results before presenting the proof.
First, we present an auxiliary lemma, where we note that the dual LP as defined in \ref{eq:dual} is degenerate up to translations: 
\begin{proposition}
	\label{prop:degen}
	For any $\Lambda$, given any $\beta$ a feasible solution to the dual LP, then $\forall c \in \mathbb{R}$, $\beta+c\mathbf{1} = \{\beta_i+c\}$ is also a feasible point with the same objective value.
\end{proposition}
\begin{proof}
	For any $i\in V$, we have $(\beta_{i}+c) - (\beta_{j}+c) = \beta_i-\beta_j \leq d_{ij}$ since $\beta$ is feasible -- hence, we have $\beta+c\mathbf{1}$ is also feasible. Furthermore, $(\beta+c\mathbf{1})^{\intercal}\Lambda = \beta^{\intercal}\Lambda + c\mathbf{1}^{\intercal}\Lambda = \beta^{\intercal}\Lambda$, as we have:
	\begin{align*}
	\sum_i \Lambda_i &= \sum_{(i,j)} \Lambda_{ji} - \sum_{(i,j)} \Lambda_{ij}= 0
	\end{align*}
	Thus both have the same objective value.
\end{proof}
As a consequence, we can add non-negativity constraints $\alpha_i\geq 0\,\forall\,i$ to \eqref{eq:dual} without affecting its value.
Now we turn to the results in Section~\ref{sec:2node}.

\begin{proposition}[Ref. Proposition~\ref{prop:2node}]
\label{appprop:2node}
	For sufficiently large $\theta$, $\exists\,\, A_1,A_2\in \mathbb{R}$ such that:
	\begin{itemize}
		\item  $\lambda = \mu \; \Rightarrow \; \gamma^{\theta} = A_1\theta^{1/2}$
		\item  $\lambda \neq \mu \; \Rightarrow \gamma^{\theta} = A_2\theta^{-1/2}e^{-\frac{\rho(\lambda-\mu)^{2}}{2(1-\rho)(\lambda+\mu)}.\theta} + o\left(\theta^{-1/2}e^{-\frac{\rho(\lambda-\mu)^{2}}{2(1-\rho)(\lambda+\mu)}.\theta}\right)$
	\end{itemize}
\end{proposition}

\begin{proof}
	First, when $\lambda \neq \mu$, we have:
	\begin{align*}
	2\gamma^{\theta} &= \theta(a-1)(\mu-\lambda)+ (a+1).\mathbb{E}[(|X-Y| + |(\mu - \lambda) - (X-Y)|] - \theta \left((a-1).(\mu - \lambda) + (a+1).|\mu - \lambda|\right)\\
	&=(a+1).\mathbb{E}[(|X-Y| + |(\mu - \lambda) - (X-Y)|]- \theta (a+1).|\mu - \lambda|
	\end{align*}
	Here we have $X-Y \sim \mathcal{N} \left({\rho\theta(\mu-\lambda), \sigma^2}\right)$, with $\sigma^2 = \sigma_X^2 + \sigma_Y^2$; hence, its absolute value is distributed as the folded normal distribution, with expected value $\xi_{\rho}$ given by: 
	\begin{align*}
	\xi_{\rho} &= \sqrt{\frac{2}{\pi}}\sigma .\mathrm{exp}\left(-\frac{\rho(\lambda-\mu)^{2}}{2(1-\rho)(\lambda+\mu)}.\theta\right) - 
	\rho\theta|\lambda-\mu|.\textrm{erf}\left(\frac{-\rho\theta|\lambda-\mu|}{\sqrt{2\theta.\rho(1-\rho)(\lambda+\mu)}}\right)
	\end{align*}
	Similarly, $(\mu - \lambda) - (X-Y)\sim \mathcal{N} \left({(1-\rho)\theta(\mu-\lambda), \sigma^2}\right)$, which has a similar expression for its expectation $\xi_{1-\rho}$. Next, to simplify the expression for $\xi$, we substitute the following asymptotic approximation of the error function:
	\begin{align*}
	\textrm{erf}(x)=1-\frac{e^{-x^2}}{\sqrt{\pi}x}+\frac{e^{-x^2}}{2\sqrt{\pi}x^3}+o\left(\frac{e^{-x^2}}{x^3}\right)
	\end{align*}
	Substituting, we get that for any normal distribution with mean $m$ and variance $s$, the expectation $\xi$ of the corresponding folded normal distribution is given by: 
	\begin{align*}
	\xi &= \sqrt{\frac{2}{\pi}}.s.e^{-\frac{m^{2}}{2s^2}} + |m|\mathrm{erf}\left(\frac{|m|}{\sqrt{2}s}\right)
	= |m| + \sqrt{\frac{2}{\pi}}\frac{s^3}{m^2}.e^{-\frac{m^{2}}{2s^2}} + o\left(\frac{s^3}{m^2}.e^{-\frac{m^{2}}{2s^2}}\right)
	\end{align*}
	Substituting in the earlier expressions, and recalling that we assume $\rho \leq 1/2$, we get that $\exists\, \; (A, \tilde{A})$ such that:
	
	\begin{align*}
	\gamma^{\theta} &= A.\theta^{-1/2}.e^{-\frac{\rho(\lambda-\mu)^{2}}{2(1-\rho)(\lambda+\mu)}.\theta} + \tilde{A}.\theta^{-1/2}.e^{-\frac{(1-\rho)(\lambda-\mu)^{2}}{2\rho(\lambda+\mu)}.\theta}  + o\left(\theta^{-1/2}.e^{-\frac{\rho(\lambda-\mu)^{2}}{2(1-\rho)(\lambda+\mu)}.\theta}\right)\\
	&=A_2\theta^{-1/2}.e^{-\frac{\rho(\lambda-\mu)^{2}}{2(1-\rho)(\lambda+\mu)}.\theta} + o\left(\theta^{-1/2}.e^{-\frac{\rho(\lambda-\mu)^{2}}{2(1-\rho)(\lambda+\mu)}.\theta}\right),
	\end{align*}
	where $A_2$ is an appropriately chosen constant. This proves the second result.
	
	On the other hand, suppose that $\lambda=\mu$. In this case, we have $\gamma^{\theta} = (a+1)\mathbb{E}[|X-Y|]$, and we can directly use the expression for the mean of the folded normal distribution with mean $0$ and variance $\sigma^2$ to get: 
	\begin{align*}
	\gamma^{\theta} = (a+1)\sqrt{\frac{2{\rho(1-\rho)(\lambda+\mu)}}{\pi}}\cdot\sqrt{\theta}
	\end{align*}
	Thus we have that $\exists\, \; A_1$ such that $\gamma^{\theta} = A_1\theta^{1/2}$.
\end{proof}

\begin{theorem}[Ref. Theorem \ref{thm:nonuniform}]
\label{appthm:nonuniform}
	Given network $G$, demand vector $\Lambda$ and demand-split vector $\rho$, assume the following hold:
	\begin{itemize}
		\item $\Lambda \in C_{\eta}$ for some $\eta \in \mathcal{E}$ (i.e., the monopolist demand vector lies in the interior of a dual cone)
		\item There is a r.v. $Z$ such that $\Lambda^{\theta,a} = \theta\rho\odot\Lambda + \sqrt{\theta}Z$ with $\mathbb{E}[Z_e]=0\,\forall\,e\in E$ and $\mathbb{E}[||Z||_1]<\infty$
		\item $\forall\,e\in E$, we have $\mathbb{P}[|Z_e|>t]=\mathcal{O}(f(t))$ with $f(t)=\mathcal{O}(t^{-1})$
	\end{itemize}
	Now suppose we define: 
	$$ L=\alpha^\intercal\rho\odot\Lambda+\beta^\intercal(1-\rho)\odot\Lambda-\eta^\intercal\Lambda$$
	Then we have: 
	\begin{itemize}
		\item $\rho\odot\Lambda \in \mathring{C}_{\alpha} \; and \; (1-\rho)\odot\Lambda \in \mathring{C}_{\beta}\Rightarrow \gamma^{\theta}= L\theta +\mathcal{O}(\theta f(\sqrt{\theta}))$
		\item $\rho\odot\Lambda \in C_{\alpha}\setminus\mathring{C}_{\alpha} \; or \; (1-\rho)\odot\Lambda \in C_{\beta}\setminus\mathring{C}_{\beta} \Rightarrow \gamma^{\theta} = L\theta + \Theta(\theta^{1/2})$
	\end{itemize}
\end{theorem}
\begin{proof}
	The proof is the same as the homogeneous case, except that each companies' upper and lower bounds have to be treated separately. 
	
	Suppose we are in the first case. We can define $\delta$ in the first case and use concentration inequalities as before. In particular, we have for the first firm that $RC(\Lambda^{\theta,a})=\alpha_1^\intercal \Lambda^{\theta,a}$, and moreover:
	\begin{align*}
	\mathbb{E}[\alpha_1^\intercal \Lambda^{\theta,a}]&=\mathbb{E}\left[\alpha^\intercal \Lambda^{\theta,a}\mathds{1}_{\{\forall e\in E, \, Z_e \leq \sqrt{\theta}\delta\}}\right] + 
	\mathbb{E}\left[\alpha_1^\intercal \Lambda^{\theta,a}\mathds{1}_{\{\exists\, e\in E \mbox{ s.t.} \, Z_e > \sqrt{\theta} \delta\}}\right]\\
	&\leq \theta\alpha^\intercal\rho\odot\Lambda + 
	\Big( \theta \bar{\beta}||\Lambda||_1 + \bar{\beta}\sqrt{\theta}\mathbb{E}\left[\ || Z ||_{1} \; |\;\exists\, e\in E \mbox{ s.t.} \, Z_e >\sqrt{\theta}\delta\right]\Big)
	\mathbb{P}\left[\exists\, e\in E \mbox{ s.t.} \, Z_e > \sqrt{\theta}\delta\right]\\
	&=\theta\alpha^\intercal\rho\odot\Lambda+\mathcal{O}\left(\theta f(\sqrt{\theta})\right)
	\end{align*}
	The upper bounds use the exact same argument as before. Thus we get a fast decay for one company, and if we sum each term of $\gamma^{\theta}$ then we get the desired result for the first case.
	\hfill \break
	For the second case, we only need only one out of two companies to have a square root decay. Let's assume that the company $b$ is under the affected regime. Concerning the lower bound, the cost for firm $a$ can be lower bounded by its limit since it is suboptimal. Thus we have the following bounding box for any company that is under the fragmentation-resilient regime, thanks to the previous point:
 $$ \theta\alpha^\intercal \rho\odot\Lambda + \mathcal{O}\left(\theta f(\sqrt{\theta})\right) \geq \mathbb{E}[\alpha_1^\intercal \Lambda^{\theta,a}]\geq \mathbb{E}[\alpha^\intercal \Lambda^{\theta,a}]=\theta\alpha^\intercal \rho\odot\Lambda $$ 
 
 Concerning firm $b$, the argument from Theorem~\ref{thm:sqroot} can be repeated. Let's assume there are only 2 active corner points (we can adapt the proof as mentioned in \ref{thm:sqroot}) denoted $\beta$ and $\tilde{\beta}$. We have for the lower bound:
	\begin{align*}
	\mathbb{E}[\alpha_2^\intercal \Lambda^{\theta,b}]
	&\geq \mathbb{E}\big[(\tilde{\beta}^\intercal\Lambda^{\theta,b}).\mathds{1}_{\tilde{\beta}^\intercal Z \geq \beta^\intercal Z}\,+
	(\beta^\intercal\Lambda^{\theta,b}).\mathds{1}_{\beta^\intercal Z \geq \tilde{\beta}^\intercal Z}\big]\\
	&= \theta\beta^\intercal(1-\rho)\odot\Lambda 
	+\theta^{1/2}.\mathbb{E}\left[(\tilde{\beta}^\intercal Z).\mathds{1}_{\tilde{\beta}^\intercal Z \geq \beta^\intercal Z} + (\beta^\intercal Z) \left(1 - \mathds{1}_{\tilde{\beta}^\intercal Z \geq \beta^\intercal Z}\right)\right]\\
	&= \theta\beta^\intercal(1-\rho)\odot\Lambda +
	\theta^{1/2}.\mathbb{E}\left[((\tilde{\beta}-\beta)^\intercal Z).\mathds{1}_{\tilde{\beta}^\intercal Z \geq \beta^\intercal Z} + (\beta^\intercal Z) \right]\\
	&= \theta\beta^\intercal(1-\rho)\odot\Lambda +
	\theta^{1/2}.\mathbb{E}\left[((\tilde{\beta}-\beta)^\intercal Z).\mathds{1}_{\tilde{\beta}^\intercal Z \geq \beta^\intercal Z} \right]\\
	&= \theta\beta^\intercal(1-\rho)\odot\Lambda + \Omega(\theta^{1/2})
	\end{align*}
	
Concerning the higher bound of the fragmented-affected company, we consider the radius $\delta$ of a ball strictly included in $C_{\beta}\cup C_{\tilde{\beta}}$ and write:
	\begin{align*}
	\mathbb{E}[\alpha_2^\intercal \Lambda^{\theta,b}]
	&=\mathbb{E}\Big[(\alpha_2^\intercal \Lambda^{\theta,b})\times\left(\mathds{1}_{\{\forall\, e\in E,\,  Z_e \leq \sqrt{\theta}\delta\}} +\mathds{1}_{\{\exists\, e\in E \mbox{ s.t.} \, Z_e > \sqrt{\theta}\delta\}}\right)\Big]\\
	&\leq \mathbb{E}\left[(\alpha_2^\intercal \Lambda^{\theta,b}) \mathds{1}_{\{\forall\,e\in E,\, Z_e \leq \sqrt{\theta} \delta\}}\right] +\mathcal{O}(\theta f(\sqrt{\theta}))\\
	&\leq \mathbb{E}\left[(\alpha_2^\intercal \Lambda^{\theta,b})\right] +\mathcal{O}(\theta f(\sqrt{\theta}))\\
	&= \theta\beta^\intercal(1-\rho)\odot\Lambda + \theta^{1/2} \mathbb{E}\left[(\beta+\tilde{\beta})^\intercal Z)\right] +\mathcal{O}(\theta f(\sqrt{\theta}))\\
	&= \theta\beta^\intercal(1-\rho)\odot\Lambda + \mathcal{O}(\theta^{1/2})
	\end{align*}

  Finally, what remains to do is sum over both company, subtract the cost of the monopoly regime and we get the desired conclusion.
\end{proof}

\begin{theorem}
	\label{thm:anycompany}
	Given network $G$, demand vector $\Lambda$, $n$ companies and demand-split vector $(\rho_1,...,\rho_n)$ such that their sum is the vector $\mathds{1}_n$, assume the following hold:
	\begin{itemize}
		\item $\Lambda \in C_{\eta}$ for some $\eta \in \mathcal{E}$ (i.e., the monopolist demand vector lies in the interior of a dual cone)
		\item There is a r.v. $Z$ such that, $\forall i, \; \Lambda^{\theta,i} = \theta\rho_i\odot\Lambda + \sqrt{\theta}Z^i$ with $\mathbb{E}[Z_e^i]=0\,\forall\,e\in E$ and $\mathbb{E}[||Z^i||_1]<\infty$
		\item $\forall\,e\in E$, we have $\mathbb{P}[|Z_e^i|>t]=\mathcal{O}(f(t))$ with $f(t)=\mathcal{O}(t^{-1})$
		\item $ \forall i, \; \rho_i\odot\Lambda \in \mathcal{C}_{\alpha_i}$
	\end{itemize}
	Now suppose we define: 
	$$ L=\left( \sum_{i=1}^n \alpha_i^\intercal \rho_i\odot\Lambda \right) - \eta^\intercal\Lambda$$
	Then we have: 
	\begin{itemize}
		\item $\forall i,\; \rho_i\odot\Lambda \in \mathring{C}_{\alpha} \; \Rightarrow \gamma^{\theta}= L\theta +\mathcal{O}(\theta f(\sqrt{\theta}))$
		\item $\exists i, \; \rho_i\odot\Lambda \in C_{\alpha}\setminus\mathring{C}_{\alpha} \Rightarrow \gamma^{\theta} = L\theta + \Theta(\theta^{1/2})$
	\end{itemize}
\end{theorem}
\begin{proof}
The previous proof gives bounds for each company depending on the underlying regime of their expected demand. We make distinction by denoting $\alpha_i$ the corner point of the expected demand, and $\xi_i$ the corner point of the random variable $\Lambda^{\theta,i}$. We have for the resilient and affected regime respectively:
\begin{align*}
\theta\alpha_i^\intercal \rho_i\odot\Lambda + \mathcal{O}\left(\theta f(\sqrt{\theta})\right) \geq \mathbb{E}[\xi_i^\intercal \Lambda^{\theta,i}]\geq \theta\alpha_i^\intercal \rho_i\odot\Lambda 
\\
\theta\alpha_i^\intercal \rho_i\odot\Lambda + \mathcal{O}(\theta^{1/2}) \geq \mathbb{E}[\xi_i^\intercal \Lambda^{\theta,i}]\geq \theta\alpha_i^\intercal \rho_i\odot\Lambda + \Omega(\theta^{1/2})
\end{align*}
 Thus depending on the regime, we get different bounds that can then be summed to obtain this result.
\end{proof}

\section{Rebalancing and Stochastic Dynamics}
\label{appsec:rebalancing}

As we mention in Section~\ref{sec:setting}, the formulation of the rebalancing cost in terms of the min-cost circulation problem arises in many models of MoD systems. We now discuss three different justifications given for it in literature,
\begin{enumerate} 
\item Best-case scenarios: For given $\Lambda_{ij}$, the total number of trip requests between pairs of nodes $(i,j)$, $RC(\Lambda)$ represents the \emph{minimum ex-post transportation cost}, i.e., the minimum cost of rebalancing empty vehicles to meet the demand. This is particularly relevant for fragmentation-affected settings, as the true cost in such cases must grow as fast as the LP.
\item Long-term planning models: These arise in settings where the platform knows all trip requests in advance, e.g., reservation models like Lyft Shuttle. They are also valid to an extent in settings where passengers have high patience~\cite{ZP14,PSFR12,SGF15}.

In more detail: suppose trip requests between pairs of nodes $i,j$ arrive according to a Poisson process with rate $\Lambda_{ij}$. Then, for an MoD platform with $m$ vehicles serving $n$ stations, to satisfy all demands, the number of vehicles must satisfy the following {\em stability condition} (via a simple time-conservation argument):  
\begin{align*}
m > \sum_{(i,j)\in E}\Lambda_{ij}\tau_{ij} + RC(\Lambda),
\end{align*}
When passengers have infinite patience (i.e., are willing to tolerate any finite delay), then it is not hard to show that the condition is also a sufficient condition -- in particular, the system is stabilized by a natural static probabilistic rebalancing policy (cf. \cite{TPF13}). 

The argument for settings with reservations is similar, except that the trip planning can now be done in advance rather than in an online manner.

\item Closed-queueing networks and passenger loss models: This is the opposite scenario from the second one, where passengers are assumed to be very impatient, and leave from the system if no vehicle is available nearby. In more detail, consider a setting with Poisson arrivals and $0$ patience (i.e., where passengers leave the system immediately if there is no available vehicle). Note that in this setting, the solution to the LP $RC(\Lambda)$ suggests a natural probabilistic rebalancing policy, where we rebalance every vehicle in the ratio of the rebalancing flow to the total flow of vehicles from that node. 

Given the passenger loss model, it is clear that any policy must incur some demand loss. Somewhat surprisingly, however, Banerjee et al.~\cite{BFL17} show that the probabilistic rebalancing policy is near-optimal under the large-market scaling, even after accounting for this demand loss. In particular, they prove that as long as $\sum_{i,j}\Lambda_{ij}\tau_{ij} = o(m)$, then static rebalancing with probabilities chosen via the LP in \eqref{eq:primal} incurs a loss which is at most a $(1+n/m)$ fraction of $RC(\Lambda)$. 

To translate this to our setting, note that all our results consider PoF scaling with respect to the demand, and are independent of the number of vehicles. For the fragmentation-affected setting, our results still hold since $RC(\Lambda)$ is a lower bound on the true costs. On the other hand, the result of~\cite{BFL17} shows that the true cost under probabilistic rebalancing is only off from $RC(\Lambda)$ by a factor of $(1+n/m)$. Now, as long as $m$ scales slightly faster than the demand scaling, we still get vanishing PoF in the fragmentation-resilient setting.
\end{enumerate}

\section{Adversarial Price of Fragmentation}
\label{appsec:advpof}

Parallel to the study of the stochastic PoF, we analyze the worst-case scenario - which we define as the \textit{Adversarial Price of Fragmentation} - to check if the increase due to competition is negligible or not. One way to do so is to consider rebalancing costs over all feasible demand splits, i.e.: 
\begin{align}
& \max \quad
RC(\Lambda^a) \; + \; RC(\Lambda - \Lambda^a) \label{eq:worstpof}\\
& \mathrm{s.t.} \qquad 0 \leq \Lambda_{ij}^a \leq \Lambda_{ij}\quad\forall\,(i,j)\in E \nonumber
\end{align}
This captures the worst-case demand-split in terms of rebalancing costs, assuming firms act myopically to minimize their costs. 
Though such a formulation is appealing for empirical analysis, it however leads to an intractable problem; in particular, as a consequence of Proposition~\ref{prop:convex}, we have that the function $RC(\Lambda^a) + RC(\Lambda - \Lambda^a)$ is convex in $\Lambda^a$, and hence maximizing it over a polyhedron is in general computationally hard.
 
To circumvent this, we use the following heuristic: First, since the polyhedron on which we maximize over is bounded and the cost function is convex, we deduce that the worst case split is reached on a corner point of the polyhedron. This polyhedron is isomorphic to $[0,1]^{N^{2}}$, so this maximization problem is equivalent to an ILP in $\{0,1\}^{N^{2}}$. For this, we utilize a projected subgradient heuristic for the dual problem, taking advantage of the cost function being continuous and convex, and differentiable almost everywhere. In particular, note that we can write the dual as:
$ \underset{\alpha \in \mathcal{E}}{\max} \; \alpha^\intercal \Lambda
$, where $\mathcal{E}$ denotes the -finite- set of corner points; hence it is differentiable except on degenerate values of $\Lambda$. 
We obtain an admissible value of the sub gradient via a first order Taylor expansion between two points of a plane, and update our solution to a new corner point depending on the sign of the gradient.
Let $f(\kappa),\,\kappa \in [0,1]^{N^2}$ denote the PoF under the demand-split $\Lambda^a = \{\kappa_{i,j}\Lambda_{i,j}\}$. 

\noindent\textbf{Adversarial PoF for NYC Data}:
Even though the above heuristic does not guarantee optimality at convergence, using it on the NYC data, we obtain a fixed point corresponding to a rebalancing cost increase of 567\% compared to the monopoly (for the demand data between noon and 1 pm on May 10th). This observation shows that demand fragmentation can lead to large systemic losses. Furthermore, note that when considering the worst case loss for a higher number of companies across which the demand may fragment, we know that the results can only be worse (i.e. the duopoly gives a lower bound).

\section{Additional Plots}
\label{appsec:plots}

In this appendix, we present some additional plots. In Figure.~\ref{fig:cluster20}, \ref{fig:cluster40}, \ref{fig:cluster60} and \ref{fig:cluster80}, we show the locations of the cluster centers, for different choices of number of stations.

\begin{figure*}[h!]
	\centering
	\subfigure[Inhomogeneous Market-Shares]{
		\label{fig:3comp}
		\includegraphics[width=0.45\columnwidth]{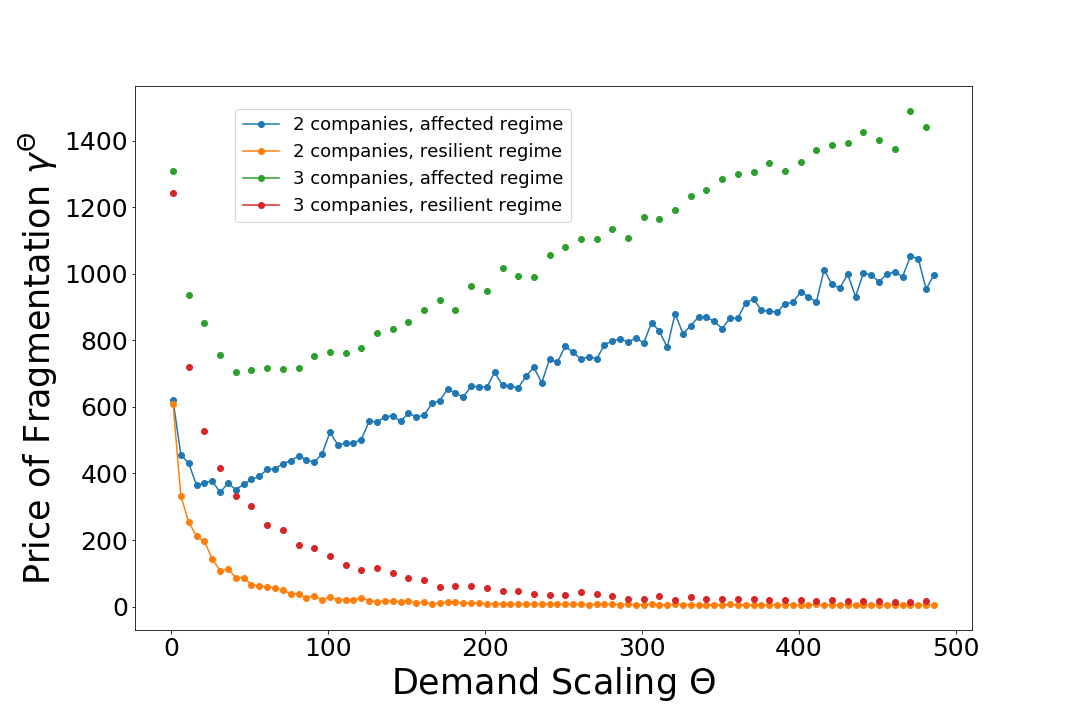}
	}
	\hspace{1cm}
	\subfigure[Effect of Multiple Firms]{
		\includegraphics[width=0.45\columnwidth]{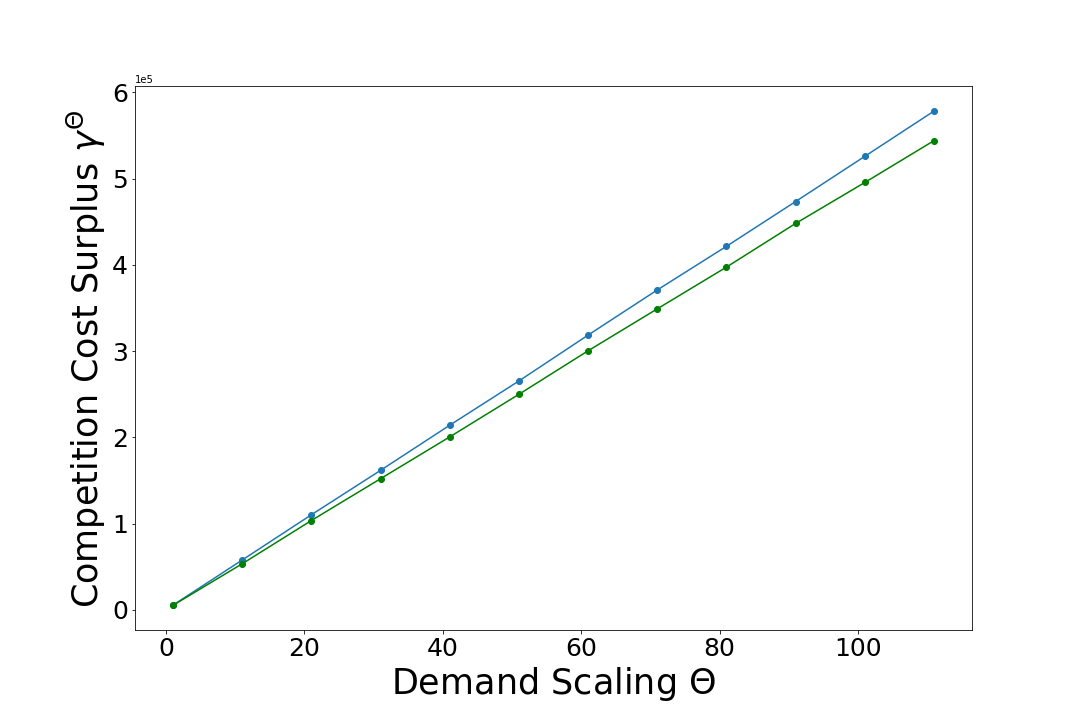}
		\label{fig:inhomogen}
	}
	\caption[Other settings of PoF scaling]{Multiple firms and heterogeneous market-shares: In Fig.~\ref{fig:inhomogen}, we demonstrate the PoF scaling with heterogeneous market shares (chosen uniformly at random on each edge). The plot clearly demonstrates the linear scaling of PoF with $\theta$. Note that the y-axis is in multiples of $10^5$; this shows that inhomogeneous demand-splits (and therefore adversarial) can have much higher PoF. Similarly, in Fig.~\ref{fig:3comp}, we demonstrate the change in PoF scaling when going from $2$ to $3$ firms; note that the situation gets worse with $3$ firms, but qualitatively, the scaling behavior is similar. Both plots are based on the same demand distribution data as in \ref{fig:intro}.}
	\label{fig:othercases}
\end{figure*}

\begin{figure*}[h!]
	\centering
	\subfigure[20 stations]{
		\includegraphics[width=0.45\columnwidth]{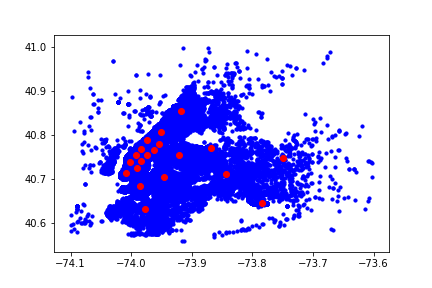}
		\label{fig:cluster20}
	}
	\hspace{1cm}
	\subfigure[40 stations]{
		\label{fig:cluster40}
		\includegraphics[width=0.45\columnwidth]{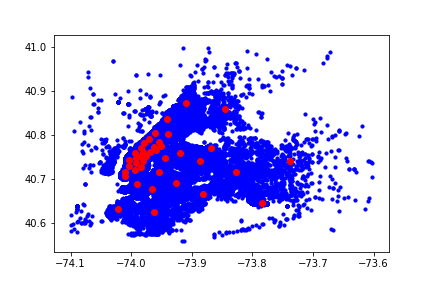}
	}
	
	\subfigure[60 stations]{
		\label{fig:cluster60}
		\includegraphics[width=0.45\columnwidth]{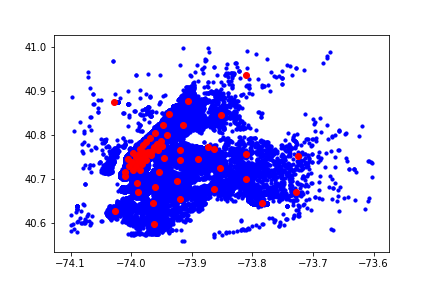}
	}
	\hspace{1cm}
	\subfigure[80 stations]{
		\label{fig:cluster80}
		\includegraphics[width=0.45\columnwidth]{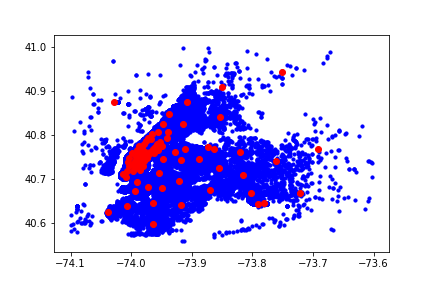}
	}
	\caption[Demand clustering]{Scatter plot of the demand locations with the position of clusters for increasing number of stations.}
	\label{fig:cluster}
\end{figure*}


Next, we simulate and plot the PoF for inhomogeneous market shares, using the same data as \ref{fig:intro}, with demand-split ratios picked uniformly at random for each edge. The splits we picked result in a non-zero scaling coefficient $L$ (defined in Theorem~\ref{thm:nonuniform}; note that our theorem predicts a linear divergence, which is confirmed by the plot in Fig.~\ref{fig:inhomogen}. This suggests that spatially inhomogeneous demand-splitting can lead to much higher costs. Finally, in Figure.~\ref{fig:3comp}, we compare the two fragmentation regimes when the number of firms is either 2 or 3. Note that the phase transition is qualitatively the same (though quantitatively worse)  with more firms. 
The data for both plots is the same data as \ref{fig:intro}.



\end{document}